\documentclass{article}
\usepackage{authblk}
\usepackage{amsmath}%
\usepackage{amsfonts}%
\usepackage{amssymb}%
\usepackage{graphicx}
\usepackage[tworuled]{algorithm2e} 
\usepackage{xspace}
\usepackage[all]{xy}

\newtheorem{theorem}{Theorem}[section]

\newtheorem{corollary}[theorem]{Corollary}
\newtheorem{lemma}[theorem]{Lemma}

\newtheorem{question}[theorem]{Question}

\newtheorem{definition}{Definition}[section]

\newenvironment{proof}[1][Proof]{\textbf{#1.} }{\ \rule{0.5em}{0.5em}}

\DeclareMathOperator{\diam}{diam}
\DeclareMathOperator{\ext}{bond}
\DeclareMathOperator{\vol}{vol}

\DeclareMathOperator{\linecenter}{lines}

\newcommand{\Hag}[1]{B^+_{\geq #1}}
\newcommand{\Hal}[1]{B^+_{\leq #1}}
\newcommand{\Hbg}[1]{B^-_{\geq #1}}
\newcommand{\Hbl}[1]{B^-_{\leq #1}}

\DeclareMathOperator{\id}{id}
\newcommand{\catname}[1]{\ensuremath{\mathbf{#1}}\xspace}
\newcommand{\catAb}{\catname{Ab}}
\newcommand{\catSet}{\catname{Set}}
\newcommand{\catomega}{\catname{\omega}}
\newcommand{\PP}{\ensuremath{\mathcal{P}\mkern-12mu\mathcal{P}}\xspace}
\newcommand{\naturalto}{%
  \mathrel{\vbox{\offinterlineskip
    \mathsurround=0pt
    \ialign{\hfil##\hfil\cr
      \normalfont\scalebox{1.2}{.}\cr
      $\longrightarrow$\cr}
  }}%
}

\begin{document}

\title{Sandpile monomorphisms and limits}
\author[a]{Moritz Lang\footnote{ To whom correspondence should be addressed. E-mail: moritz.lang@ist.ac.at}}
\author[a]{Mikhail Shkolnikov}
\affil[a]
{Institute of Science and Technology Austria, 
Am Campus 1, 3400 Klosterneuburg, Austria}%

\date{\today}
\maketitle

\begin{abstract}
We introduce a tiling problem between bounded open convex polyforms $\hat{P}\subset\mathbb{R}^2$ with directed and uniquely colored edges. If there exists a tiling of the polyform $\hat{P}_2$ by $\hat{P}_1$, we show that one can construct a monomorphism from the sandpile group $G_{\Gamma_1}=\mathbb{Z}^{\Gamma_1}/\Delta(\mathbb{Z}^{\Gamma_1})$ on the domain (graph) $\Gamma_1=\hat{P}_1\cap\mathbb{Z}^2$ to the respective group on $\Gamma_2=\hat{P}_2\cap\mathbb{Z}^2$. We provide several examples of infinite series of such tilings with polyforms converging to $\mathbb{R}^2$, and thus the first definition of scaling-limits for the sandpile group on the plane. Additional results include an exact sequence relating sandpile configurations to harmonic functions, an alternative formula for the order of the sandpile group based on a basis for the module of integer-valued harmonic functions, and three examples of how to prove the existence of (cyclic) subgroups for infinite families of sandpile groups by constructing appropriate integer-valued harmonic functions.
The main open question concerns if the scaling-limits of the sandpile group for different sequences of polyforms converging to $\mathbb{R}^2$ are isomorphic.
\end{abstract}

\section{Introduction}
\subsection{Background}
Let $\bar{\Gamma}=\Gamma\cup\{s\}$ be the vertices of a finite connected (multi-)graph with sink $s$. Denote by $\partial\Gamma$ the boundary of $\Gamma$ -- the set of all vertices adjacent to the sink. The standard discrete graph Laplacian $\bar{\Delta}_{\bar\Gamma}$ is then defined as the difference between the adjacency matrix of $\bar\Gamma$ and its degree/valency matrix. When we delete, from $\bar{\Delta}_{\bar\Gamma}$, the row and column corresponding to the sink, we obtain the reduced graph Laplacian $\Delta_\Gamma$.
The sandpile group $G_\Gamma$ is then defined as the cokernel of $\Delta_\Gamma$ acting on $\mathbb{Z}^\Gamma$ \cite{Dhar1990,Biggs1999,Raza2014,Alar2017}, i.e.
\begin{align*}
G_\Gamma=\mathbb{Z}^\Gamma/\Delta_\Gamma(\mathbb{Z}^\Gamma).
\end{align*}
Note that the sandpile groups corresponding to different choices of the sink for the same graph $\bar{\Gamma}$ are isomorphic \cite{Cori2000}. Also note that the sandpile group was rediscovered several times and that it is, as a consequence, sometimes referred to as the critical group--based on the work of Biggs \cite{Biggs1999,Biggs1999b}--or as the Jacobian and (sometimes) as the Picard group \cite{Bacher1997,Biggs1997,Baker2007}.

The study of the sandpile group originated in the physical literature, and there mainly focuses on sandpile groups defined on finite connected domains of the standard square lattice $\mathbb{Z}^2$, i.e. on graphs $\bar\Gamma$ obtained from $\mathbb{Z}^2$ by contracting all vertices $\mathbb{Z}^2\setminus(\mathbb{Z}^2\cap P)$ outside of some finite open set $P\subset\mathbb{R}^2$ to the sink.
The group naturally arises in the study of the sandpile model, a cellular automaton introduced by Bak, Tang and Wiesenfeld in 1987 \cite{Bak1987} as the first and archetypical example of a system showing self-organized criticality (SOC), a phenomenom which subsequently became important in several areas of physics, biology, geology and other fields (see \cite{Aschwanden2016} for a recent review). Shortly after the introduction of this cellular automaton, Dhar showed that its recurrent configurations form a group isomorphic to $G_\Gamma$, and laid the foundation for its analysis \cite{Dhar1990,Dhar1995}. 
Due to this isomorphism, both $G_\Gamma$ as well as the group formed by the recurrent configurations of the sandpile model are commonly referred to as the sandpile group. 
This may cause some confusion, since the elements of $G_\Gamma$ rather correspond to the equivalence classes of recurrent configurations.
For readers used to the notation of the literature on the sandpile model, we thus note that the distinction between transient and recurrent configurations does not apply when directly working with $G_\Gamma$. We also note that we denote the group operation by $+:G_\Gamma\times G_\Gamma\rightarrow G_\Gamma$, and not by $(.+.)^\circ$, with $(.)^\circ:\mathbb{Z}_{\geq 0}^\Gamma\rightarrow\{0,\ldots,3\}^\Gamma$ the relaxation operator \cite{Dhar1990}. 

The sandpile group, specifically when defined on domains of $\mathbb{Z}^2$, provides connections between various mathematical fields, including fractal geometry, graph theory and algebraic geometry (see below), tropical geometry \cite{Caracciolo2010,Kalinin2016,Kalinin2017,Kalinin2018b,Kalinin2019}, domino tilings \cite{Florescu2015}, and others. 
Via the so called ``burning algorithm'', Dhar constructed bijections (in the category of sets) between the sandpile group and spanning trees \cite{Dhar1990,Majumdar1992}, and thus showed that the former is a refinement of the latter. 
Creutz was the first to study the recurrent configuration of the sandpile model corresponding to the identity on domains of $\mathbb{Z}^2$, and provided an iterative algorithm for its construction \cite{Creutz1990}. He found that, on many domains, this identity is composed of self-similar fractal patterns \cite{Creutz1990,LeBorgne2002,Holroyd2008,Caracciolo2008,Caracciolo2008b}; since these patterns appear to be remarkably similar on rectangular domains with the same aspect ratio, scaling limits for the sandpile identity have been conjectured \cite{Holroyd2008}. Recently, we have extended these conjectures and suggested that several scaling limits for each recurrent configuration exist, forming piecewise smooth ``fractal movies'' referred to as harmonic sandpile dynamics \cite{Lang2019}.
Finally, based on an analogy of graphs and (discrete) Riemann surfaces established by Baker and Norine (including a Riemann-Roch theorem for graphs) \cite{Baker2007,Baker2009}, several connections between algebraic geometry and the study of sandpiles were established (see \cite{Perkinson2011} and \cite{Corry2018}, p.65ff. and 191ff., for expositions). The most interesting connection, in the context of this article, is provided by (non-constant) harmonic morphisms between graphs, corresponding to holomorphic maps between surfaces, as these morphisms directly induce epimorphisms between the respective sandpile groups \cite{Baker2009} (see also \cite{Reiner2014,Alfaro2012}).

Only for few infinite families of graphs, the structure of the respective sandpile groups has been (partly) determined, including complete graphs \cite{Lorenzini1991}, complete multipartite graphs \cite{Jacobson2003}, cycles (equivalent to domains of $\mathbb{Z}^1$) \cite{Lorenzini1991}, thick cycles \cite{Alar2017} (see also \cite{Dhar1995}), wheels \cite{Biggs1999, Norine2011}, modified wheels \cite{Raza2014}, wired regular trees \cite{Levine2009}, thick trees \cite{Chen2007}, polygon flowers \cite{Chen2019}, nearly complete graphs \cite{Norine2011}, threshold graphs \cite{Norine2011}, M\"obius ladders \cite{Pingge2006,Deryagina2014}, prism graphs/graphs $\mathcal{D}_n$ of the dihedral group \cite{Dartois2003,Deryagina2014}, and $n$-cubes \cite{Bai2003,Alfaro2012}.
While this list is certainly not complete, the decomposition of the sandpile group on domains of $\mathbb{Z}^2$, for which the sandpile model was originally defined \cite{Bak1987}, is--to our knowledge--yet unknown. Numeric calculations of the order \cite{Dhar1990} or decomposition \cite{Dhar1995} of the sandpile group on small enough domains however indicate that the groups are in general ``incompatible'', even when the domains have the same shape, in the sense that no group monomorphisms can exist between them. For example, the order of the sandpile group on a $3\times 3$ square domain is $2^{11}7^2$, while the one on a $5\times 5$ domain is $2^{18}3^55^211^213^2$. 
Recently, we have shown that the sandpile group can be considered as a discretization of a $|\partial\Gamma|$-dimensional torus, to which we refer to as the extended sandpile group \cite{Lang2019}. We have then derived epimorphisms from the extended sandpile group on a given domain to the corresponding group on a subdomain. On the level of the (usual) sandpile group, due to the discretization, this renormalization is however defined in the category of sets and in general only ``approximates'' group homeomorphism for sufficiently large domains \cite{Lang2019}. Under which conditions these ``approximations'' can be lifted to true group homeomorphisms, if possible at all, is yet unknown. 
This lack of known relationships in terms of homeomorphism between the sandpile groups on different domains of $\mathbb{Z}^2$ is in stark contrast to the role of the sandpile model as the archetypical example for self-organized criticality, given that the concept of criticality itself is based on the notion of scaling. Progress in this subject might also provide means by which the existence of the conjectured scaling limits of the sandpile identity \cite{Holroyd2008} and of other recurrent configurations \cite{Lang2019} might be proven.

\subsection{Sandpile monomorphisms (main result)}
In this paper, we analyze the relationships between sandpile groups defined on different domains of the standard square lattice $\mathbb{Z}^2$. Specifically, given two domains $\Gamma_1, \Gamma_2\subset\mathbb{Z}^2$, $\Gamma_1\subseteq \Gamma_2$, our goal is to understand under which conditions group monomorphisms from $G_{\Gamma_1}$ to $G_{\Gamma_2}$ exist. 

\begin{figure}[tb!]
	\centering
	 	\includegraphics[width=0.9\textwidth]{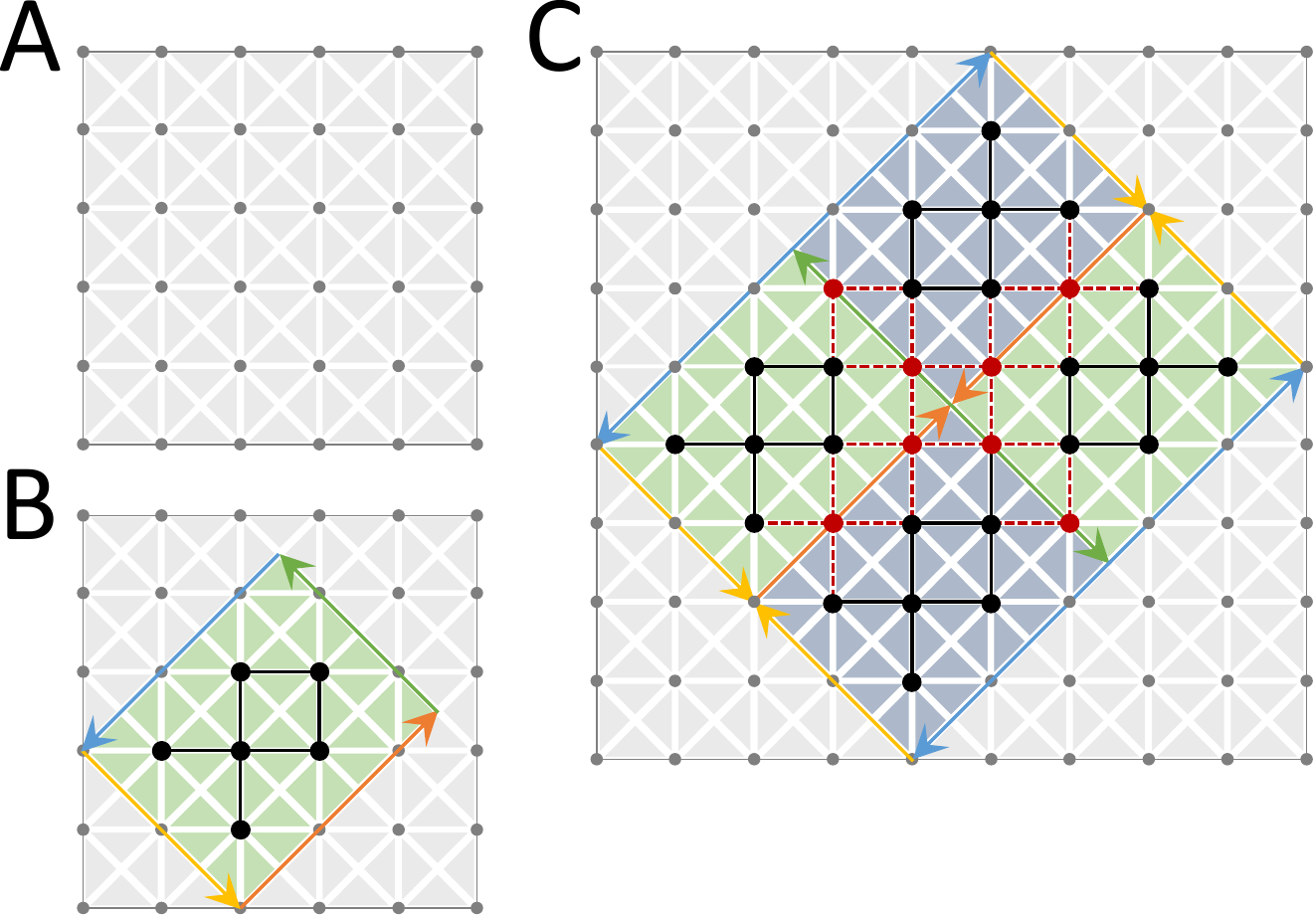}
	\caption{A) Dark-gray points represent the vertices of the standard square lattice $\mathbb{Z}^2$, while the gray isoscele triangles correspond to $M$.
	B) The isocele triangles belonging to the $M$-polyform $P_1$ are highlighted by a green background. The black points and lines represent the vertices and edges of the graph $\Gamma(P_1)=\mathbb{Z}^2\cap P_1$ defined by $P_1$. The sides of the $M$-polyform are directed and colored, exemplifying the definition of $P_1^{DC}$. 
	C) DC-tiling of a $M$-polyform $P_2$ (all colored isoscele triangles) by four copies of the $M$-polyform $P_1^{DC}$ from (B). The background of the tiles are colored green if they can be obtained from $P_1^{DC}$ by only translations and rotations, and blue if (additionally) reflections are required. Note that the graph $\Gamma(P_2)$ consists not only of the vertices and edges corresponding to the four tiles (black points and lines), but also of additional vertices lying on the common edges of pairs of tiles, and the edges connecting these vertices to the rest of the graph (red points and lines).
	}
	\label{fig:initialIllus}
\end{figure}
To state our main result, we first introduce some notation.
Let $M$ be the unique tiling of $\mathbb{R}^2$ by isosceles triangles with base length $1$ and height $\frac{1}{2}$ such that each vertex of $(\mathbb{Z}+0.5)^2$ coincides with the apecies of four triangles (Figure~\ref{fig:initialIllus}A). An $M$-polyform $P\subset M$ then consists of a finite connected subset of triangles in $M$ (Figure~\ref{fig:initialIllus}B). Note that $M$ is \textit{not} the usual triangular tiling of the plane, that the corners and edges of its triangles do not form a lattice, and that $M$-polyforms thus differ from the usual definition of polyiamonds.

By a slight abuse of notation, we interpret each $M$-polyform $P$ to directly correspond to the open subset of $\mathbb{R}^2$ enclosed by its isosceles triangles, i.e. to the interior of $\bigcup_{m\in P}m$.
To each $M$-polyform $P$, we then associate the domain $\Gamma(P)=\mathbb{Z}^2\cap P$. This domain is obtained from the standard square lattice $\mathbb{Z}^2$ as described in the Introduction, i.e. by contracting all vertices $\mathbb{Z}^2\setminus(\mathbb{Z}^2\cap P)$ to the sink (Figure~\ref{fig:initialIllus}B). We interchangeably denote by $G_\Gamma$ and $G_P$ the sandpile groups defined on the domain $\Gamma=\Gamma(P)$.

Denote by $P^{DC}$ the result of assigning directions and colors to the edges of an $M$-polyform $P$ such that each edge has a different color (Figure~\ref{fig:initialIllus}B). Given two $M$-polyforms $P_1$ and $P_2$, we say that $P_1$ DC-tiles $P_2$ if there exists a tiling $T^{P_1\rightarrow P_2}$ of $P_2$ by copies of $P_1^{DC}$ (allowing all transformations which correspond to automorphisms of $M$), such that every common edge of two adjacent tiles in $T^{P_1\rightarrow P_2}$ has the same color and direction (Figure~\ref{fig:initialIllus}C).

We can now state our main result:
\begin{theorem}\label{theorem:main}
Let $P_1$ and $P_2$ be two convex $M$-polyforms, and assume that $P_1$ DC-tiles $P_2$. Then, there exists a group monomorphism $G_{P_1}\rightarrowtail G_{P_2}$ from $G_{P_1}$ to $G_{P_2}$.
\end{theorem}
In the proof of this theorem, we construct an explicit mapping $\mu(T^{P_1\rightarrow P_2})=(G_{P_1}\rightarrowtail G_{P_2})$
from $DC$-tilings to the corresponding sandpile group monomorphisms. We refer to Section~\ref{proof:main} for the details on the construction of this map, and here only discuss some of its properties. We note that the graph morphisms $\Gamma(P_2)\rightarrow\Gamma(P_1)$ induced by DC-tilings are in general not harmonic at the sink $s_2$ of $\Gamma(P_2)$ (Figure~\ref{fig:initialIllus}B\&C), and that thus Theorem~\ref{theorem:main} is distinct from the theory on harmonic graph morphisms \cite{Baker2009}. 

Trivially, for two $M$-polyforms $P_1$ and $P_2$, there can exist more than one distinct DC-tiling of $P_2$ by $P_1$. For example, let the polyform $P$ describe a square with width $w$ and sides parallel to the standard axes of $\mathbb{R}^2$. Since the dihedral group $D_4$ of a square has order eight, there also exist eight different DC-tiling of $P$ by itself. For $w>2$, $\mu$ maps each of these tilings to a different automorphism of $G^{P}$, which directly correspond to the action of the respective element of $D_4$ on $\Gamma(P)$ (see proof of Theorem~\ref{theorem:main}). For $w=2$, the domain $\Gamma(P)$ however consists of only a single vertex, and all eight tilings are mapped to the trivial automorphism. 
Now, denote by $\hat{P}$ the result of extending a polyform $P$ by one triangle in $M$ adjacent to $P$ such that $\Gamma(\hat{P})=\Gamma(P)$. Then, there exist no DC-tilings of $\hat{P}$ by $P$, or vice versa. However, since $G^P=G^{\hat{P}}$, the set of automorphisms is non-empty. We thus conclude that the mapping $\mu$ is in general neither injective nor surjective.

\subsection{Scaling-limits of the sandpile group}
Let $\PP$ denote the poset of bounded convex $M$-polyforms, with $P_1\subseteq_{DC} P_2$ if there exists a DC-tiling $T^{P_1\rightarrow P_2}$ of the $M$-polyform $P_2$ by the $M$-polyform  $P_1$ such that 
the position and orientation of one tile in $T^{P_1\rightarrow P_2}$ directly corresponds to $P_1^{DC}$, i.e. $P_1^{DC}\in T^{P_1\rightarrow P_2}$. We naturally identify $\PP$ with its corresponding (small) category, with the (faithful) forgetful functor $U:\PP\rightarrow\catSet$ to the category of sets mapping each $M$-polyform to its corresponding open subset of $\mathbb{R}^2$ and $\subseteq_{DC}$ to set inclusions.

Since the position and orientation of one tile uniquely identifies a $DC$-tiling (if it exists), the definition of $\PP$ allows us to associate a $DC$-tiling $\nu(P_1\subseteq_{DC} P_2)\in\{T^{P_1\rightarrow P_2}\}$ to each morphism $P_1\subseteq_{DC} P_2$, i.e. the unique DC-tiling satisfying $P_1^{DC}\in T^{P_1\rightarrow P_2}$. We can then define the functor $F:\PP\rightarrow\catAb$ from $\PP$ to the category $\catAb$ of abelian groups, with $F(P)=G_P$ and $F(P_1\subseteq_{DC} P_2)=\mu(\nu(P_1\subseteq_{DC} P_2))$. To see that $F(\id_P)=\id_{G_P}$ and $F((P_2\subseteq_{DC} P_3)\circ (P_1\subseteq_{DC} P_2))=(G_{P_2}\rightarrowtail G_{P_3})\circ (G_{P_1}\rightarrowtail G_{P_2}$), we refer to the construction of the map $\mu$ in Section~\ref{proof:main}.

\begin{figure}[tb!]
	\centering
	 	\includegraphics[width=0.8\textwidth]{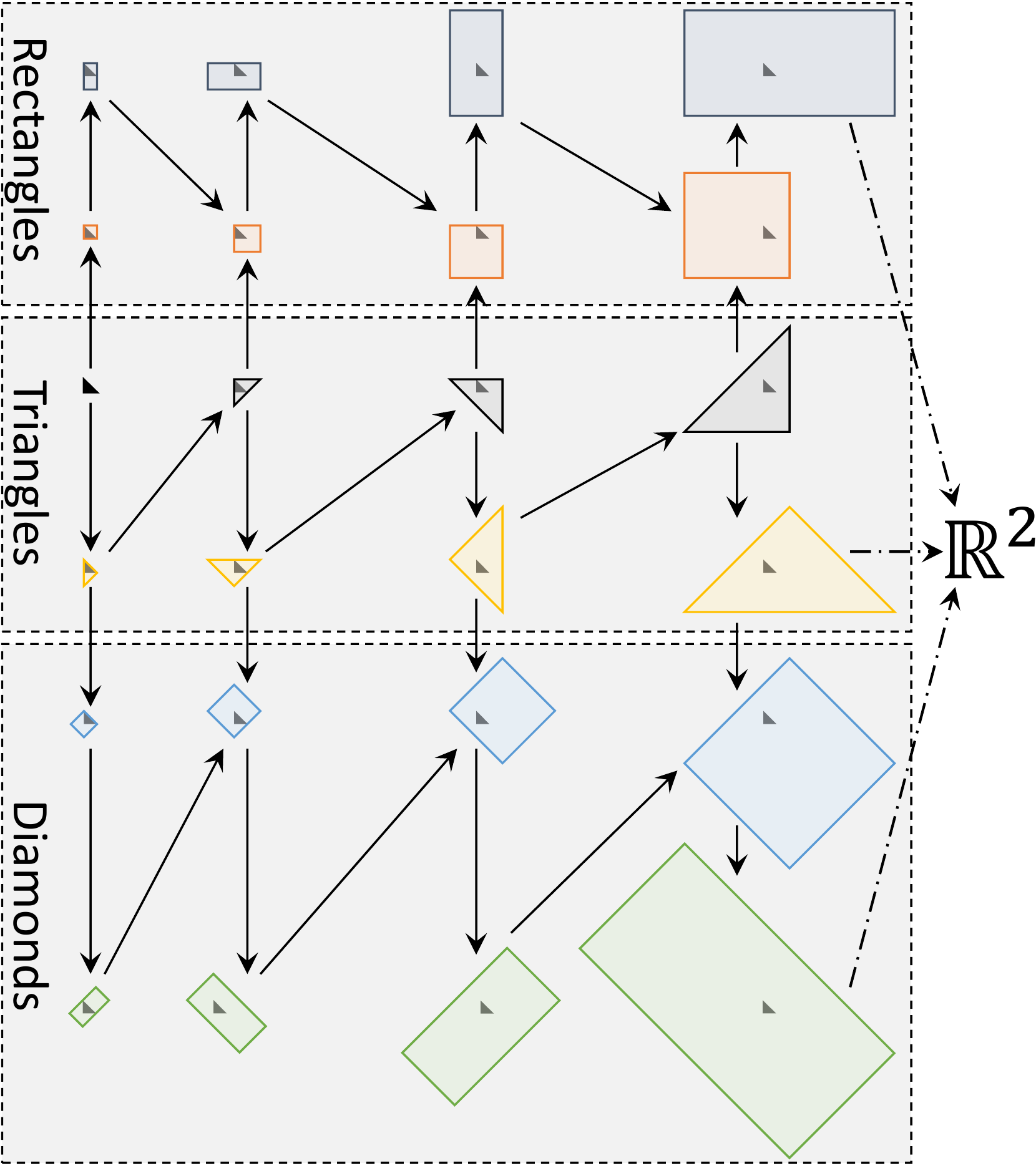}
	\caption{Depiction of a small finite part of the category \PP of $M$-polyforms. Each shape represents a $M$-polyform $P$, while arrows represent morphisms $P_1\subseteq_{DC}P_2$ (identities and composed morphisms omitted). For a better orientation, the position of the initial triangular $M$-polyform (black) is depicted by a gray background in each $M$-polyform. The category \PP is not filtered, since there exist no $DC$-tiling of diamond-shaped $M$-polyforms by rectangular-shaped ones, or vice versa. However, both classes of $M$-polyforms can be reached by triangular-shaped ones. Note that the composition of all non-bounded sequences $S$ with the forgetful functor $U$ in the depicted part of $\PP$ has a direct limit of $\mathbb{R}^2$.
	}
	\label{fig:category}
\end{figure}

Of specific interest are infinite sequences $S=S_0 \subseteq_{DC} S_1 \subseteq_{DC} S_2\ldots$ of $M$-polyforms in $\PP$ (identity and composed morphisms omitted), i.e. functors $S\in\PP^\catomega$ from the usual linear order $\catomega=\{0,1,\ldots\}$ on the ordinal numbers to $\PP$. Trivially, each of these sequences, composed with the forgetful functor $U$, defines a direct limit $\varinjlim U S=\bigcup_i U(S_i)\subseteq\mathbb{R}^2$ (in the category of sets, since \PP does not admit all filtered colimits), which we denote by $\hat{S}_\infty$. Furthermore, each sequence, composed with $F$, also defines a direct limit $\varinjlim F S$, denoted either as $G_{\hat{S}_\infty}^S$ or, equivalently, by $G_{\Gamma(\hat{S}_\infty)}^S$.
We interpret $G_{\Gamma(\hat{S}_\infty)}^S$ as the limit of the sandpile group for $\Gamma(S_i)\rightarrow\Gamma(\hat{S}_\infty)$ (with respect to the sequence $S$). In Figure~\ref{fig:category}, we depict the morphisms between four families of polyforms in $\PP$. The direct limit of each infinite sequence $S\in\PP^\catomega$ which only contains these polyforms and morphisms, with $S_{i+1}\neq S_i$ for all $i\in\catomega$, is given by $\hat{S}^\infty=\mathbb{R}^2$, and thus $\Gamma(\hat{S}^\infty)=\mathbb{Z}^2$. To our knowledge, the respective limits of the sandpile group $G_{\mathbb{Z}^2}^S$ are the first\footnote{This claim of priority might be controversial: there exist several approaches to define sandpile models directly on $\mathbb{Z}^2$, which cope with the occurrence of infinite avalanches in various ways. The (weak) limits of the sandpile measures for some of these models can be associated to/concentrate on certain abelian groups, see e.g. \cite{Maes2004,Athreya2004,Maes2005,Jarai2015} and references therein. At least to us, it is however unclear if and how these groups exactly relate to the categorical notion of (scaling) limits for the sandpile group employed in this article.} definitions of scaling limits for the sandpile group on $\mathbb{Z}^2$.

If a given sequence $S$ of $M$-polyforms is upper bounded, i.e. if there exists an $u\in\catomega$ such that $U S_j=U S_u=\hat{S}_\infty$ for all $j\geq u$, it directly follows that $G_{\hat{S}_\infty}^S\cong G_{S_u}$. Thus, for such upper bounded sequences, the limit of the sandpile group is completely determined (up to isomorphisms) by the upper bound, i.e. $F$ preserves all finite direct limits. In such cases, we can drop the dependency of $G_{\hat{S}_\infty}^S$ on $S$ and simply write $G_{\hat{S}_\infty}$. We may ask if the same also holds for unbounded sequences:
\begin{question}\label{question:limitPreservation}
Let $S^A,S^B\in\PP^\catomega$ be two (possibly unbounded) sequences of $M$-polyforms with common limit $\hat{S}_\infty=\varinjlim U S^A=\varinjlim U S^B$. Is $G_{\hat{S}_\infty}^{S^A}$ isomorphic to $G_{\hat{S}_\infty}^{S^B}$?
\end{question}
Let $\hat{\PP}$ be the category with objects corresponding to all limits $\hat{S}_\infty=\varinjlim U S$ of sequences $S\in\PP^\catomega$ of polyforms, and morphisms $\hat{S}_\infty^A\subseteq \hat{S}_\infty^B$ if there exists a natural transformation $S^a\naturalto S^b$ between two sequences $S^a,S^b\in\PP^\catomega$ with $\hat{S}^A_\infty=\varinjlim U S^a$ and $\hat{S}^B_\infty=\varinjlim U S^b$, i.e. if $S^a_i\subseteq_{DC}S^b_i$ for all $i\in\catomega$. 
We interpret $\PP$ to represent a full subcategory of $\hat{\PP}$, with the object function of the (fully faithful) inclusion functor $I:\PP\rightarrow\hat{\PP}$ given by $I(P)=\varinjlim U\delta P$, where $\delta:\PP\rightarrow\PP^\catomega$ denotes the diagonal functor with $(\delta P)_i=P$ for all $i\in\catomega$. Question~\ref{question:limitPreservation} then asks if there exists a functor $\hat{F}:\hat{P}\rightarrow\catAb$ which preserves all direct limits, and for which $F$ factors as $\hat{F}\circ U$.

In case Question~\ref{question:limitPreservation} can be answered in the affirmative, the dependency of the direct limit of the sandpile group on the sequence could be always dropped. Specifically, this would mean that there exists a unique scaling limit $G_{\mathbb{Z}^2}$ (up to isomorphisms) of the sandpile group on $\mathbb{Z}^2$.
In this case, we would however immediately arrive at the following result:
\begin{corollary}\label{equivalenceLimits}
Assume that $G_{\hat{S}_\infty}^{S^A}\cong G_{\hat{S}_\infty}^{S^B}$ whenever $\hat{S}_\infty=\varinjlim U S^A=\varinjlim U S^B$. Then, the limit of the sandpile group on $\mathbb{Z}^2$ is isomorphic to its limit on the upper-right quadrant of $\mathbb{Z}^2$, i.e. $G_{\mathbb{Z}^2}\cong G_{\mathbb{Z}_{\geq 0}^2}$.
\end{corollary}
This corollary, as well as several similar ones relating the limits of the sandpile group on different unbounded domains, arises because the mapping $\nu$ between morphisms $P_1\subseteq_{DC}P_2$ and DC-tilings $T^{P_1\rightarrow P_2}$ is not injective. We can thus construct two sequences $S^A$ and $S^B$ such that there exists a natural isomorphism $F S^A\cong F S^B$, but for which $\hat{S}^A_\infty\neq\hat{S}^B_\infty$. Corollary~\ref{equivalenceLimits} then follows when choosing $S^A_0=S^B_0$ to be square-shaped $M$-polyforms with side length $w_0$, $S^A_{i+1}$ and $S^B_{i+1}$ to have side lengths $w_{i+1}=5w_i$, $S^A_{i+1}$ to be positioned such that $S^A_{i}$ is in its center, and $S^B_{i+1}$ such that $S^B_{i}$ is at its bottom-left.

\subsection{An exact sequence and the order of the sandpile group}
Theorem~\ref{theorem:main} is based on a close relationship between sandpile groups and certain modules of harmonic functions. Since this relationship is of interest itself, we summarize some of its properties in this section. 
We say that a domain $\Gamma\subseteq\mathbb{Z}^2$ is convex if there exists a convex open set $P\subseteq\mathbb{R}^2$ such that $\Gamma=P\cap\mathbb{Z}^2$. Note that, different to before, we do not require $P$ to be an $M$-polyform anymore.
We say that an $R$-valued function $H:\Gamma\rightarrow R$, $R\in\{\mathbb{Z},\mathbb{Q},\mathbb{R}\}$, is harmonic (on $\Gamma$) if $\Delta_\Gamma H(v)=0$ for all vertices $v\in\Gamma_0$ in the interior $\Gamma_0=\Gamma\setminus\partial\Gamma$ of the domain. The $R$-valued harmonic functions on $\Gamma$ form the module $\mathcal{H}^\Gamma_R$.

\begin{lemma}\label{lemma:exactSequence}
For every finite convex domain $\Gamma\subset\mathbb{Z}^2$, $-\Delta_\Gamma:\mathcal{H}_G^\Gamma\cong G^\Gamma$ is an isomorphism from $\mathcal{H}_G^\Gamma=\{H\in\mathcal{H}_\mathbb{Q}^\Gamma|\Delta_\Gamma H|_{\partial\Gamma}\in\mathbb{Z}^{\partial\Gamma}\}/\mathcal{H}_\mathbb{Z}^\Gamma$ to the sandpile group $G_\Gamma$, with $\mathcal{H}_G^\Gamma$ the subgroup of the rational-valued harmonic functions $\mathcal{H}_\mathbb{Q}^\Gamma$ with integer-valued Laplacians, modulo the integer-valued harmonic functions $\mathcal{H}_\mathbb{Z}^\Gamma$. This isomorphism corresponds to the exact sequence
\begin{align*}
\xymatrix{
0 \ar[r] & G_\Gamma \ar[r] & \mathcal{H}_\mathbb{Q}^\Gamma/\mathcal{H}_\mathbb{Z}^\Gamma \ar[r] & (\mathbb{Q}/\mathbb{Z})^{\partial\Gamma} \ar[r] &  0.
}
\end{align*}
\end{lemma}
We derive an explicit construction for this isomorphism in the proof of Lemma~\ref{lemma:exactSequence}.

Denote by $\mathcal{B}^\Gamma_R=\{B_i\}_{i=1,\ldots,|\partial\Gamma|}$ a basis for the module $\mathcal{H}^\Gamma_R$ of $R$-valued harmonic functions on a finite convex domain $\Gamma\subset\mathbb{Z}^2$ (in Section~\ref{section:basis}, we present an algorithm for the construction of $\mathcal{B}^\Gamma_\mathbb{Z}$, and thus also for $\mathcal{B}^\Gamma_\mathbb{Q}$ and $\mathcal{B}^\Gamma_\mathbb{R}$).
By definition, the Laplacian $\Delta_\Gamma H$ of every harmonic function $H\in\mathcal{H}^\Gamma_R$, and thus also of every basis function in $\mathcal{B}^\Gamma_R$, only has support at the boundary $\partial\Gamma$ of the domain. The Laplacian of every basis function in $\mathcal{B}^\Gamma_R$ can thus be restricted to $\partial\Gamma$ without information loss, and we refer to
$\Delta\mathcal{B}^\Gamma_R=(\Delta_\Gamma B_1|_{\partial\Gamma},\ldots,\Delta_\Gamma B_{|\partial\Gamma|}|_{\partial\Gamma})\in R^{|\partial\Gamma|\times |\partial\Gamma|}$ 
as the potential matrix of $\Gamma$ (with respect to $\mathcal{B}^\Gamma_R$).
\begin{lemma}\label{lemma:order}
Let $\Gamma\subset\mathbb{Z}^2$ be a finite convex domain, and $\mathcal{B}^\Gamma_\mathbb{Z}$ be a basis for the module of integer-valued harmonic functions $\mathcal{H}^\Gamma_\mathbb{Z}$ on $\Gamma$. Then, the order of the sandpile group $G_\Gamma$ is given by
\begin{align*}
|G_\Gamma|=|\det(\Delta\mathcal{B}^\Gamma_\mathbb{Z})|.
\end{align*}
\end{lemma}

\subsection{Integer-valued harmonic functions and cyclic subgroups}
In this section, we present three examples of constructions which directly link integer-valued harmonic functions to cyclic subgroups of the sandpile group. Such constructions might help to answer Question~\ref{question:limitPreservation} (in the negative), since they represent structural restrictions on an (eventually existing) unique scaling-limit $G_{\mathbb{Z}^2}$ of the sandpile group.

\begin{figure}[bt]
	\centering
	 	\includegraphics[width=0.75\textwidth]{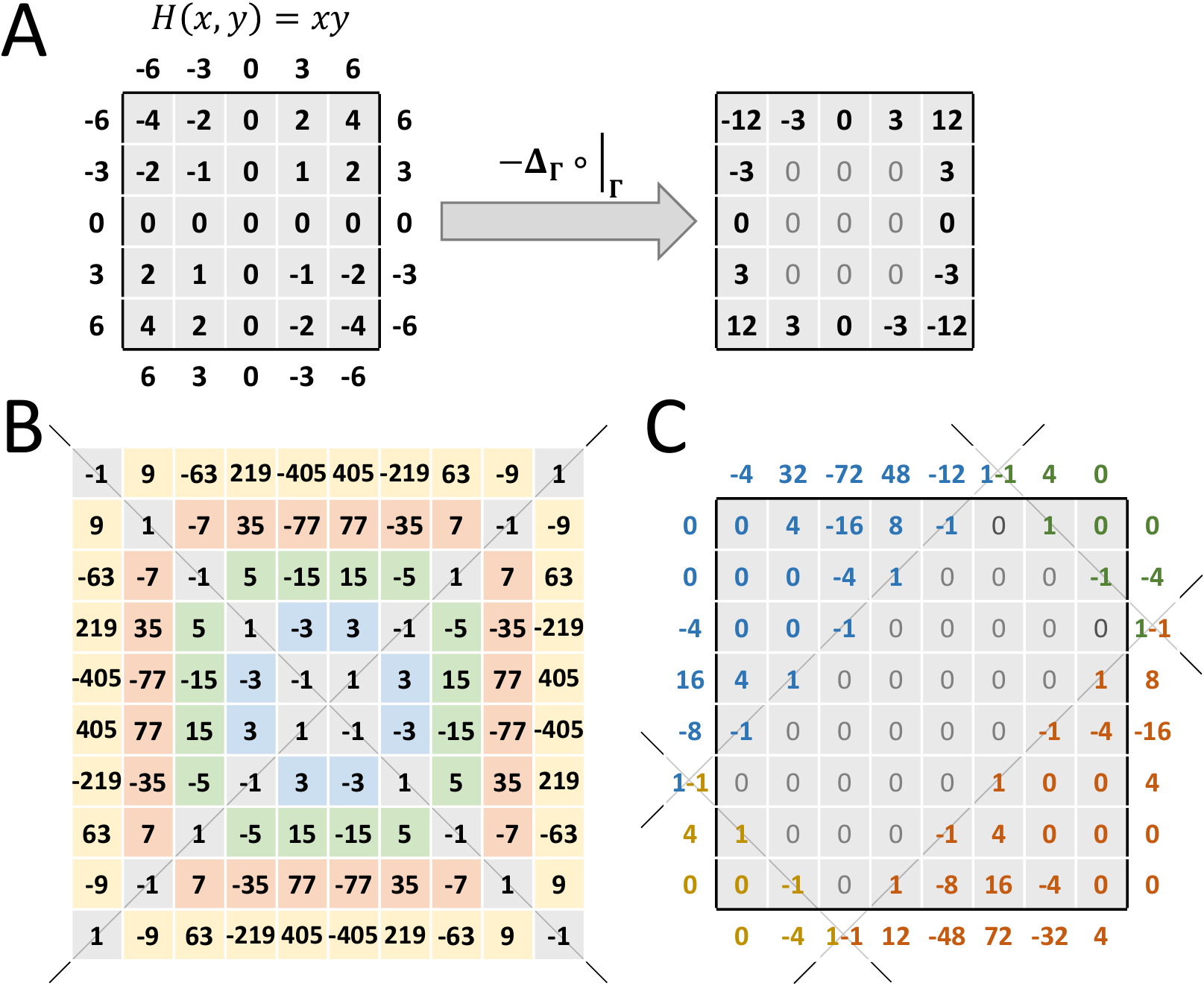}
	\caption{Integer-valued harmonic functions used to prove the existence of cyclic subgroups.
	A) The harmonic function $H=xy$ is coprime on $N\times N$ domains, $N\in2\mathbb{N}+1$, (here: $N=5$) while its Laplacian is divisible by $\frac{N+1}{2}$ (here: by $3$).
	B) All values of the harmonic function $H^{\pi}$ are divisible by three on vertices with a blue background, by five on vertices with a green background, and by seven on vertices with a red background, respectively. This sequence however ends at the vertices with a yellow background, since nine is not prime.
	C) The harmonic function $H^\diamond_i$ corresponds to the sum of the four harmonic basis functions depicted in blue, green, red, and yellow, such that the values of $H^\diamond_i$ on $\partial(\mathbb{Z}^2\setminus\Gamma)$ are all divisible by four.
	}
	\label{fig:primeHarmonic}
\end{figure}

Let $H:\mathbb{Z}^2\rightarrow\mathbb{Z}$, $\Delta_{\mathbb{Z}^2}H=0$, be an integer-valued harmonic function on $\mathbb{Z}^2$. Assume that the values of the restriction of $H$ to some finite convex domain $\Gamma\subset\mathbb{Z}^2$ are coprime, and that the values of $H$ on the boundary $\partial(\mathbb{Z}^2\setminus\Gamma)$ of the complement of the domain are all divisible by some integer $n\geq 2$ (Figure~\ref{fig:primeHarmonic}A). Recall that 
\begin{align*}
\Delta_\Gamma H|_\Gamma(v)=-\sum_{\stackrel{w\in\partial(\mathbb{Z}^2\setminus\Gamma)}{w\sim v}}H(w)
\end{align*}
for all $v\in\Gamma$, and that, thus, also $\Delta_\Gamma H|_\Gamma$ is divisible by $n$ (Figure~\ref{fig:primeHarmonic}A). From the isomorphism in Theorem~\ref{lemma:exactSequence}, it then follows that
$
S_\Gamma^H=\{0,C,\ldots(n-1)C\}\subseteq G_\Gamma
$
forms a cyclic subgroup of the sandpile group $G_\Gamma$ with generator $C=[-\frac{1}{n}\Delta_\Gamma H|_\Gamma]$ and order $|S_\Gamma^H|=n$. 

For a given integer-valued harmonic function $\hat{H}\in\mathcal{H}_{\mathbb{Z}}^\Gamma$, we recently introduced the harmonic sandpile dynamics $D_{\hat{H}}:\mathbb{R}/\mathbb{Z}\rightarrow G_\Gamma$, 
$D_{\hat{H}}(t) = [\lfloor t\Delta_\Gamma\hat{H}\rfloor]$, 
with $\lfloor.\rfloor$ the element-wise floor function \cite{Lang2019}.
For $\hat{H}=H|_\Gamma$, the discussion above implies that the elements of $S_\Gamma^H$ appear exactly at times $t\equiv0, \frac{1}{n},\ldots,\frac{n-1}{n} (\operatorname{mod} 1)$ in the harmonic sandpile dynamics $D_{\hat{H}}$, in the sense that $[\lfloor t\Delta_\Gamma H|_\Gamma\rfloor]=[t\Delta_\Gamma H|_\Gamma]\in S_\Gamma^H$ at these times.

As a first example of how such constructions can impose restrictions on the limits of the sandpile group, consider the family $\{\Gamma_N\}_{N\in2\mathbb{N}+1}$ consisting of all $N\times N$ square domains $\Gamma_N\subset\mathbb{Z}^2$ with odd domain sizes $N$. Assume that each domain $\Gamma_N$ is defined such that its center lies at the origin $(x,y)=(0,0)$ of $\mathbb{Z}^2$. It is then easy to see that the values of the harmonic function $H=x y$ on $\partial(\mathbb{Z}^2\setminus\Gamma_N)$ are divisible by $\frac{N+1}{2}$ (Figure~\ref{fig:primeHarmonic}A), which directly proves the following lemma:
\begin{lemma}\label{lemma:divFirstExample}
For every $N\in 2\mathbb{N}+1$, the sandpile group $G_{\Gamma_N}$ on an $N\times N$ square domain $\Gamma_N\subset\mathbb{Z}^2$ has a cyclic subgroup $\mathbb{Z}/\frac{N+1}{2}\mathbb{Z}\subseteq G_{\Gamma_N}$ of order $\frac{N+1}{2}$.
\end{lemma}
For every $n\geq 1$, we can construct a sequence of domains $S^n\in\PP^\catomega$ with $S^n_\infty=\varinjlim U S^n=\mathbb{R}^2$, which starts at an $M$-polyform $S^n_0$ with $\Gamma(S^n_0)$ corresponding to an $N\times N$ square domain with $N=2n-1$ (Figure~\ref{fig:category}). If Question~\ref{question:limitPreservation} can be answered in the affirmative, this would directly imply that the scaling limit $G_{\mathbb{Z}^2}$ of the sandpile group on the standard square lattice $\mathbb{Z}^2$ would contain cyclic subgroups of every order.

For $N\times N$ square domains with even domain sizes $N\in 2\mathbb{N}$, somewhat similar results can be obtained when considering the integer-valued harmonic function depicted in Figure~\ref{fig:primeHarmonic}B. This harmonic function takes values on $\partial(\mathbb{Z}^2\setminus\Gamma_N)$ which are divisible by three ($N=2$), five ($N=4$) and seven ($N=6$), respectively. This pattern however breaks down at $N=8$, since $N+1=9$ is not prime.
\begin{lemma}\label{lemma:divPrime}
Let $\Gamma_N\subset\mathbb{Z}^2$ be an $N\times N$ square domain. Then, if $N+1$ is prime, the sandpile group $G_{\Gamma_N}$ possesses a cyclic subgroup $\mathbb{Z}/(N+1)\mathbb{Z}\subseteq G_{\Gamma_N}$ of order $N+1$.
\end{lemma}

If $N+1$ is not prime, Theorem~\ref{theorem:main} implies that there exist group monomorphisms from $G_{\Gamma_M}$ to $G_{\Gamma_N}$ whenever $M+1$ divides $N+1$. Thus, the sandpile group on every $N\times N$ domain (independently if $N+1$ is prime or not) possesses a cyclic subgroup with an order given by the product of all distinct factors of $N+1$ (each to the power of one). Furthermore, when we denote by $p_k$ the $k^{th}$ prime number, the limit $G_{\mathbb{Z}^2}^S$ of the sandpile group with respect to the sequence $S\in\PP^\catomega$ with $\Gamma(S_i)=\Gamma_{\prod_{k\leq i}(p_k-1)}$ contains at least one cyclic subgroup $\mathbb{Z}/p_k\mathbb{Z}\subset G_{\mathbb{Z}^2}^S$ for every prime number $p_k$.

Our last lemma is based on the existence of integer-valued harmonic functions which are zero in the interior of some diamond-shaped region of $\mathbb{Z}^2$, take values $\pm 1$ at its edges, the value $0$ at its corners, and which are divisible by four everywhere else (Figure~\ref{fig:primeHarmonic}C). 
By the discussion above, every such harmonic function for which all four corners lie on the boundary $\partial(\mathbb{Z}^2\setminus\Gamma_N)$ of the complement of an $N\times N$ domain $\Gamma_N\subset\mathbb{Z}^2$ can be directly mapped to a cyclic subgroup $\mathbb{Z}/4\mathbb{Z}\subseteq G_{\Gamma_N}$ of order four.
If we also allow for one ``degenerated diamond'' in case $N$ is odd, there exist $N$ such harmonic functions which are linearly independent.
\begin{lemma}\label{lemma:div4}
Let $G_{\Gamma_N}$ be the sandpile group on an $N\times N$ square domain $\Gamma_N\subset\mathbb{Z}^2$, with $N\in\mathbb{N}$. Then, $G_{\Gamma_N}$ has a subgroup $S_{N}\subseteq G_{\Gamma_N}$ isomorphic to the direct sum $S_{N}\cong\bigoplus_{i=1}^N(\mathbb{Z}/4\mathbb{Z})$ of $N$ cyclic groups of order four.
\end{lemma}
We note that this result was derived before in \cite{Dhar1995}, and formed the basis for the proof that the minimal number of generators for the sandpile group on $\Gamma_N$ is $N$. Our alternative proof arguably provides more insights by directly constructing the respective cyclic subgroups of $G_{\Gamma_N}$, and can be easily generalized to other domains.

\subsubsection*{Acknowledgements}
The authors thank Nikita Kalinin for carefully reading the manuscript. 
ML is grateful to the members of the Guet group for valuable comments and support. 
MS is grateful to Tamas Hausel, Ludmil Katzarkov, Maxim Kontsevich, Ernesto Lupercio, Grigory Mikhalkin and Andras Szenes for inspiring communications.

\section{Overview of the proofs}
We first derive the two related Lemmata~\ref{lemma:exactSequence} and \ref{lemma:order}. 
Lemma~\ref{lemma:exactSequence} then allows us to restate the question on the existence of monomorphisms between sandpile groups (Theorem~\ref{theorem:main}) into a question on the existence of monomorphisms between groups of harmonic functions, which we solve by an explicit construction. We then continue to state an algorithm for the construction of a basis for the integer-valued harmonic functions on a given finite convex domain. Finally, we use this basis to prove Lemmata~\ref{lemma:divPrime}\&\ref{lemma:div4}.
We note that Corollary~\ref{equivalenceLimits} and Lemma~\ref{lemma:divFirstExample} were directly proved in the Introduction.

\section{The sandpile group and harmonic functions}
In this section, we derive the isomorphism between the harmonic functions $\mathcal{H}_G^\Gamma$ and the sandpile group $G_\Gamma$ (Lemma~\ref{lemma:exactSequence}), as well as the formula for the order of the sandpile group (Lemma~\ref{lemma:order}). We start with an observation made by Creutz about the sandpile model, namely that every recurrent configuration can be reached from the empty configuration (or any other configuration) by only adding particles to the boundary of the domain and ``relaxing'' the sandpile \cite{Creutz1990}. Recall that the elements of the sandpile group, as defined in this article, correspond to the equivalence classes of the recurrent configurations of the sandpile model (see Introduction). 
Creutz's observation can thus be restated as follows: for every element $C\in G^\Gamma$ of the sandpile group, there exist (infinitely many) functions $X\in\mathbb{Z}^{\Gamma}$ which only have support at the boundary $\partial\Gamma$ of the domain, and which satisfy that $[X]=C$, with $[.]:\mathbb{Z}^\Gamma\rightarrow G_\Gamma$ the canonical projection map to the sandpile group. For an algorithm for the construction of $X$, we refer to \cite{Creutz1990}. 

By the existence and uniqueness of solutions to the discrete Dirichlet problem on convex domains \cite{Lawler2010}, it follows that, for every such $X\in\mathbb{Z}^{\Gamma}$, there exists a unique rational-valued harmonic function $H_X\in\mathcal{H}_\mathbb{Q}^\Gamma$ with $\Delta_\Gamma H_X=-X$. 
The composition $[.]\circ-\Delta_\Gamma$ of the discrete Laplacian with the canonical projection map then maps two harmonic functions $H_{X,1},H_{X,2}\in\mathcal{H}_\mathbb{Q}^\Gamma$, $\Delta_\Gamma H_{X,1},\Delta_\Gamma H_{X,2}\in\mathbb{Z}^{\Gamma}$, to the same element of the sandpile group if and only if $-\Delta_\Gamma(H_{X,1}-H_{X,2})\in\Delta_\Gamma(\mathbb{Z}^\Gamma)$. Since both $\Delta_\Gamma H_{X,1}$ and $\Delta_\Gamma H_{X,2}$ only have support at the boundary $\partial\Gamma$ of the domain, $H_{X,1}-H_{X,2}$ is thus an integer-valued harmonic function, which concludes our proof of Lemma~\ref{lemma:exactSequence}.

We construct the inverse of $-\Delta_\Gamma:\mathcal{H}_G^\Gamma\cong G_\Gamma$ in two steps.
For every configuration $C\in G_\Gamma$, we first define the coordinates $\sigma_\Gamma:G_\Gamma\rightarrow(\mathbb{Q}/\mathbb{Z})^{\partial\Gamma}$,
\begin{align*}
\sigma_\Gamma(C)\equiv-(\Delta\mathcal{B}^\Gamma_\mathbb{Z})^{-1}X\ \ (\operatorname{mod} 1),
\end{align*}
with respect to the basis $\mathcal{B}^\Gamma_\mathbb{Z}$. Note that, for two different choices $X^\alpha,X^\beta\in\mathbb{Z}^{\Gamma}$, $-(\Delta\mathcal{B}^\Gamma_\mathbb{Z})^{-1}(X^\alpha-X^\beta)\in\mathbb{Z}^{\partial\Gamma}$, and that thus the coordinates $\sigma_\Gamma$ don't depend on the specific choice for $X$. This also implies that $\sigma_\Gamma$ correspond to toppling invariants as defined in \cite{Dhar1995}.

In the second step, we then define the function $\phi_\Gamma:(\mathbb{R}/\mathbb{Z})^{\partial\Gamma}\rightarrow\mathcal{H}_\mathbb{R}^\Gamma/\mathcal{H}_\mathbb{Z}^\Gamma$, $\phi_\Gamma(s)=\sum_{i=1}^{|\partial\Gamma|}s_iB_i$. It is easy to check that the composition $\phi_\Gamma\circ\sigma_\Gamma:G_\Gamma\rightarrow\mathcal{H}_\mathbb{R}^\Gamma/\mathcal{H}_\mathbb{Z}^\Gamma$ is independent of the choice of the basis $\mathcal{B}^\Gamma_\mathbb{Z}$, and that $-\Delta_\Gamma\phi_\Gamma(\sigma_\Gamma([X]))=[X]$. The latter implies that $\phi_\Gamma\circ\sigma_\Gamma$ is the inverse of $-\Delta_\Gamma$.

The isomorphism between the sandpile group $G_\Gamma$ and $\mathcal{H}_G^\Gamma$  proposes to consider the sandpile group as a discrete subgroup of a continuous Lie group isomorphic to $\mathcal{H}_\mathbb{R}^\Gamma/\mathcal{H}_\mathbb{Z}^\Gamma$, to which we refer to as the extended sandpile group $\tilde{G}_\Gamma$ \cite{Lang2019}. 
More precisely, the extended sandpile group is an extension of the torus $(\mathbb{R}\setminus\mathbb{Z})^{\partial\Gamma}$ by the usual sandpile group, and is defined by the exact sequence
\begin{align*}
\xymatrix{
0 \ar[r] & G_\Gamma \ar[r] & \tilde{G}_\Gamma \ar[r] & (\mathbb{R}/\mathbb{Z})^{\partial\Gamma} \ar[r] &  0.
}
\end{align*}
In terms of the sandpile model, this Lie group is obtained by allowing each vertex $b\in\partial\Gamma$ in the boundary $\partial\Gamma$ of the domain to carry a real value $\tilde{C}(b)\in[0,4)$ of particles, while each vertex $v\in\Gamma_0$ in the interior $\Gamma_0=\Gamma\setminus\partial\Gamma$ of the domain is still only allowed to carry an integer number of particles, i.e. $\tilde{C}(v)\in\{0,1,2,3\}$ (the toppling rules are kept unchanged) \cite{Lang2019}. This definition lifts $\phi_\Gamma:(\mathbb{R}/\mathbb{Z})^{\partial\Gamma}\cong \tilde{G}_\Gamma$ to a group isomorphism, and a left-inverse of the inclusion map $G_\Gamma\rightarrow\tilde{G}_\Gamma$ is given by the floor function $\lfloor.\rfloor:\tilde{G}_\Gamma\rightarrow G_\Gamma$. We thus naturally arrive at the function $f=\lfloor.\rfloor\circ-\Delta_\Gamma\circ\phi_\Gamma:(\mathbb{R}/\mathbb{Z})^{\partial\Gamma}\rightarrow G_\Gamma$, $f(s)
=-[\lfloor\sum_{i=1}^{|\partial\Gamma|} s_i \Delta_\Gamma B_i\rfloor]$, which justifies to interpret the usual sandpile group $G_\Gamma$ as the discretization of an $|\partial\Gamma|$-dimensional torus \cite{Lang2019}.
 
Due to the properties of the floor function, the preimage $f^{-1}(C)$ of an element $C\in G_\Gamma$ of the sandpile group under $f$ is connected. Denote by $\vol(f^{-1}(C))$ the volume of this preimage, with $\vol((\mathbb{R}/\mathbb{Z})^{\partial\Gamma})=1$. Since, for every $C\in G_\Gamma$, there exists a coordinate transformation $s\mapsto\tilde{s}$ such that $C$ has coordinates $\tilde{s}=0$, we get that $\vol(f^{-1}(C))=\vol(f^{-1}(\mathbf{0}))=\frac{1}{|G_\Gamma|}$ for all $C\in G_\Gamma$, with $\mathbf{0}$ the identity of the sandpile group. The preimage $f^{-1}(\mathbf{0})$ of the identity under $f$ forms a $|\partial\Gamma|$-parallelotope with edges $g_i$ given by $(\Delta\mathcal{B}^\Gamma_\mathbb{Z})g_i=e_i$, with $(e_i)_j=\delta_{ij}$ the $i^{th}$ unit vector and $\delta_{ij}$ the Kronecker delta. The volume of this parallelotope is $\vol(f^{-1}(\mathbf{0}))=|\det(\Delta\mathcal{B}^\Gamma_\mathbb{Z})^{-1}|$, and thus $|G_\Gamma|=|\det(\Delta\mathcal{B}^\Gamma_\mathbb{Z})|$, which proves Lemma~\ref{lemma:order}.

\section{Construction of sandpile monomorphisms}\label{proof:main}

\begin{figure}[tb!]
	\centering
	 	\includegraphics[width=0.6\textwidth]{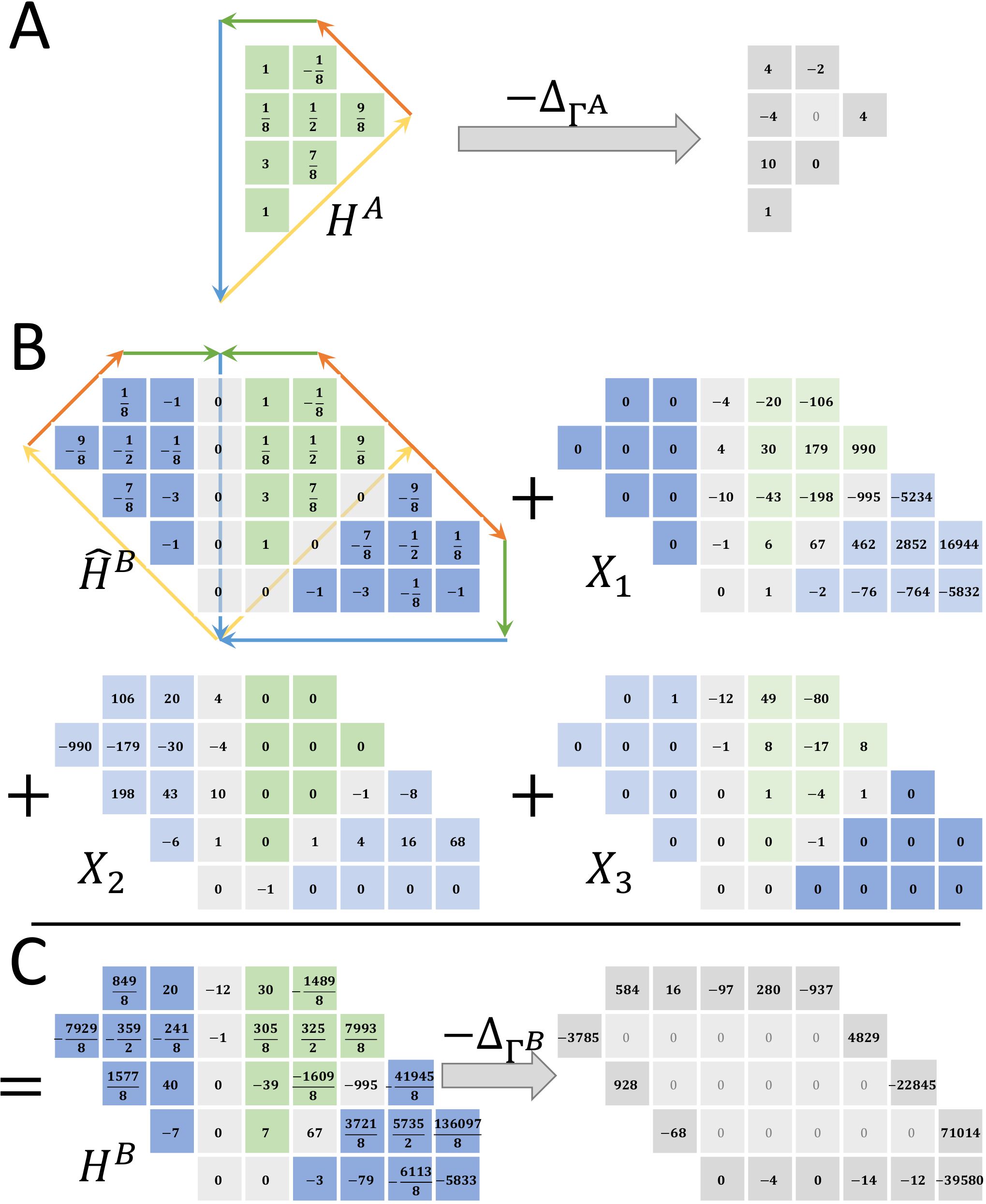}
	\caption{A) Example of an harmonic function $H_A\in\mathcal{H}_G^{\Gamma_A}$ (left) corresponding to the element $[-\Delta_{\Gamma_A}H_A]\in G_{P_A}$ of the sandpile group (right) on a given $M$-polyform $P_A$. The directed and colored edges of $P_A^{DC}$ are indicated by arrows, and the green squares corresponds to the vertices in $\Gamma_A=\Gamma(P_A)$. 
	B) The $M$-polyform $P_A$ from (A) $DC$-tiles the depicted $M$-polyform $P_B$. To construct the harmonic function $H_B\in\mathcal{H}_G^{\Gamma_B}$ onto which $H_A$  is mapped under the monomorphism from $\mathcal{H}_G^{P_A}$ to $\mathcal{H}_G^{P_B}$ induced by this tiling, we first define a rational-valued function $\hat{H}_B$ which is harmonic everywhere, except for the vertices directly next to an internal-boundary (gray backgrounds). For each tile in $T^{P_A\rightarrow P_B}$, we then construct an integer-valued function $X_i$ which cures the non-harmoniticity of its respective vertices. 
	C) The procedure depicted in (B) leads to the harmonic function $H_B\in\mathcal{H}_G^{\Gamma_B}$ (left) corresponding to the element $[-\Delta_{\Gamma_B}H_B]\in G_{P_B}$ of the sandpile group (right) on $P_B$, onto which $H_A$ is mapped by the monomorphism from $\mathcal{H}_G^{P_A}$ to $\mathcal{H}_G^{P_B}$.
	}
	\label{fig:monomorphismProof}
\end{figure}

Let $P_A$ and $P_B$ be two convex $M$-polyforms, and assume that there exists a $DC$-tiling $T^{P_A\rightarrow P_B}$ of $P_B$ by $P_A$. In this section, we then construct a monomorphism from the sandpile group on the domain $\Gamma_A=\Gamma(P_A)=\mathbb{Z}^2\cap P_A$ to the sandpile group on $\Gamma_B=\Gamma(P_B)=\mathbb{Z}^2\cap P_B$, and thus prove Theorem~\ref{theorem:main}. Before starting this construction, we derive three properties of $DC$-tilings. 

\begin{corollary}
The domains of different tiles do not overlap, i.e. $\Gamma_i\cap\Gamma_j=\{\}$ for all $i\neq j$. 
\end{corollary}
\begin{proof}
This corollary directly follows from each tile $P_i$ being treated as an open subset of $\mathbb{R}^2$ when determining its domain  $\Gamma_i=\mathbb{Z}^2\cap P_i$.
\end{proof}

The domains of the tiles in general don't cover $\Gamma_B$. Specifically, all vertices of $\Gamma_B$ which lie directly on common edges (including their endpoints) of two tiles are not elements of any $\Gamma_i$ (red vertices in Figure~\ref{fig:initialIllus}). We refer to the set $\partial^T\Gamma_B=\Gamma_B\setminus\bigcup_i\Gamma_i$ of these vertices as the internal boundaries of the tiling. These internal boundaries separate the domains $\Gamma_i$ in the following sense:
\begin{corollary}\label{corollary:separation}
The removal of all vertices in $\partial^T\Gamma_B$ splits $\Gamma_B$ into the disconnected components $\{\Gamma_1, \Gamma_2,\ldots,\Gamma_{|T^{P_A\rightarrow P_B}|}\}$.
\end{corollary} 
\begin{proof}
Because the tiles are convex by assumption, each pair of adjacent tiles can be separated by exactly one line (the extension of their common edge). By the definition of $M$, this line is either horizontal, vertical or diagonal, and passes through infinitely many vertices of $\mathbb{Z}^2$ (Figure~\ref{fig:initialIllus}A). In each of the cases, it splits $\mathbb{Z}^2$ into two unconnected components, from which the corollary directly follows.
\end{proof}

By definition, each tile $P_i\in T^{P_A\rightarrow P_B}$ can be obtained from $P_A^{DC}$ by a combination of translations, rotations and reflections. If this is possible by  using only translations and rotations, we assign the sign $s(P_i)=+1$ to the tile, and otherwise the sign $s(P_i)=-1$. This definition also induces signs $s(v)=s(\Gamma_i)=s(P_i)$ for the domains $\Gamma_i$ and vertices $v\in\Gamma_i$ belonging to the tiles. To the internal boundaries $\partial^T\Gamma_B$ and their vertices $b\in\partial^T\Gamma_B$, we assign the sign $s(\partial^T\Gamma_B)=s(b)=0$.
The relationship of each tile $P_i$ with the polyform $P_A$ (i.e. the translations, rotations and reflections mapping $P_A$ on $P_i$) corresponds to a function $\psi_i:\Gamma_A\rightarrow\Gamma_B$ which maps vertices $v_A\in \Gamma_A$ of the polyform onto their corresponding vertices $v_i$ of the tile. For two vertices $v,w\in\Gamma_B$, we then define the equivalence relation $\equiv_{DC}$ such that $v\equiv_{DC}w$ if there exists a $v_A\in \Gamma_A$ such that $v=\psi_i(v_A)$ and $w=\psi_j(v_A)$ for some tiles $P_i$ and $P_j$, or if both vertices are part of the internal boundaries, i.e. $v,w\in\partial^T\Gamma_B$. We denote by $[v]_{DC}$ the equivalence class of $v$ induced by $\equiv_{DC}$.
\begin{corollary}\label{corollary:IBs}
Let $b\in\partial^T\Gamma_B$ be a vertex of the internal boundaries. Then, the number of neighbors of $b$ in every equivalence class $[v_B]_{DC}$, $v_B\in\Gamma_B$ carrying a positive sign is equal to the number of neighbors carrying a negative sign, i.e. $\sum_{\stackrel{v\in[v_B]_{DC}}{v\sim b}}s(v)=0$.
\end{corollary}
\begin{proof}
Assume that $b$ has at least one neighbor in $[v_B]_{DC}$; otherwise the corollary is trivially satisfied. Also, assume $v_B\notin\partial^T\Gamma_B$, since otherwise $s(v)=0$ for all $v\in[v_B]_{DC}$, from which the corollary also trivially follows.
Denote by $N(b,v_B)=\{v\in[v_B]_{DC}|v\sim b\}$ the set of neighbors of $b$ in the equivalence class of $v_B$.
Every vertex can have maximally four neighbors, thus $|N(b,v_B)|\leq 4$.
Being part of the internal boundaries, $b$ must lie on at least one common edge (including endpoints) of two tiles $P_i\neq P_j$. These two tiles can be mapped onto one another by reflection on the common edge, and thus must have opposite signs. If a vertex $v\in\Gamma_i$ of $P_i$ is a neighbor of $b$, it follows that there must be a vertex $w\in\Gamma_j$ of $P_j$ which is also a neighbor of $b$, and which has opposite sign, i.e. $s(w)=-s(v)$. This excludes the case $|N(b,v_B)|=1$, and proves the corollary for $|N(b,v_B)|=2$.
For $|N(b,v_B)|\in\{3,4\}$, the structure of $M$ directly implies that $b$ has to lie on a common corner of three, respectively four, tiles. The corresponding internal angles of the tiles have to be smaller or equal to $360^\circ/3=120^\circ$, respectively $360^\circ/4=90^\circ$. The definition of $M$ only admits internal angles which are multiples of $45^\circ$ (Figure~\ref{fig:initialIllus}A). Thus, in both cases, only angles of $45^\circ$ or $90^\circ$ are possible. An angle of $45^\circ$ is only possible if all $v\in[v_B]$ lie on the internal boundaries, which implies $s(v)=0$ (see above). If the angle is $90^\circ$, $|N(b,v_B)|=3$ would imply that $P_B$ is not convex, which can thus be excluded. Finally, if the angle is $90^\circ$ and $|N(b,v_B)|=4$, each of the four tiles to which these vertices belong must have exactly two adjacent tiles with opposite signs, from which the corollary follows.
\end{proof}
 
With this preparatory work, we can now prove Theorem~\ref{theorem:main}. By Lemma~\ref{lemma:exactSequence}, the sandpile group $G_\Gamma$ is isomorphic to $\mathcal{H}_G^\Gamma=\{H\in\mathcal{H}_\mathbb{Q}^\Gamma|\left.\Delta_\Gamma H\right|_{\partial_\Gamma}\in\mathbb{Z}^{\partial\Gamma}\}/\mathcal{H}_\mathbb{Z}^\Gamma$. It thus suffices to construct a monomorphism from $\mathcal{H}_G^{P_A}$ to $\mathcal{H}_G^{P_B}$ whenever the $M$-polyform $P_A$ $DC$-tiles $P_B$. We construct this monomorphism in two steps. 
For the first step, assume that $T^{P_A\rightarrow P_B}$ is a given $DC$-tiling of $P_B$ by $P_A$, and let $H_A$ be a harmonic function in $\mathcal{H}_G^{\Gamma_A}$ (Figure~\ref{fig:monomorphismProof}A). Then, define the rational-valued function $\hat{H}_B\in\mathbb{Q}^{\Gamma_B}$ in the following way (Figure~\ref{fig:monomorphismProof}B): for each vertex $v\in\Gamma_i$ belonging to tile $P_i\in T^{P_A\rightarrow P_B}$, set $\hat{H}_B(v)=s(P_i)H_A(v_A)$ with $v_A\in\Gamma_A$ the unique vertex satisfying $\psi_i(v_A)=v$. Otherwise, that is if $v$ belongs to the internal boundaries, set $\hat{H}_B(v)=0$.

Because $\hat{H}_B(b)=0$ for all vertices $b\in\partial^T\Gamma_B$ of the internal boundaries, Corollary~\ref{corollary:separation} implies that $\Delta_{\Gamma_B}\hat{H}_B(v)=s(\Gamma_i)\Delta_{\Gamma_A} H_A(v_A)$ for all vertices $v\in\Gamma_i$ belonging to the domain of a tile $P_i$, with $\psi_i(v_A)=v$. This implies that the Laplacian of $\hat{H}_B$ is zero in the interior of the sub-domains $\Gamma_i$ of $\Gamma_B$, and integer-valued at their boundaries.
From Corollary~\ref{corollary:IBs}, on the other hand, it follows that $\Delta_{\Gamma_B}\hat{H}_B(b)=0$ for every vertex $b\in\partial^T\Gamma_B$ of the internal boundaries. Thus, $\hat{H}_B$ is harmonic nearly everywhere, except at the vertices directly adjacent to (but not including) the internal boundaries, for which $\Delta_{\Gamma_B}\hat{H}_B$ is integer-valued.

The ``harmonic deficit'' of $\hat{H}_B$ can be cured, one tile at a time: for a given tile $P_i$, we can define an integer-valued function $X_i\in\mathbb{Z}^{\Gamma_B}$ whose Laplacian is zero everywhere in the interior of $\Gamma_B$, except for those vertices $\{v\in\partial\Gamma_i\setminus\partial\Gamma_B|\exists b\in\partial^T\Gamma_B:v\sim b\}$ at the boundary of $\Gamma_i$ which are adjacent to at least one vertex of the internal boundaries, for which we require that $\Delta_{\Gamma_B}X_i(v)=-\Delta_{\Gamma_B}\hat{H}_B(v)$. For example, if we set $X_i$ to zero in $\Gamma_i$, $X_i$ directly corresponds to a solution $\tilde{X}_i$ of a Dirichlet problem on $\mathbb{Z}^2\setminus\Gamma_i$ with boundary conditions chosen such that $\sum_{\stackrel{w\in\partial(\mathbb{Z}^2\setminus\Gamma_i)}{w\sim v}}\tilde{X}_i(w)=-s(P_i)\Delta_{\Gamma_B}\hat{H}_B(v)$ for each $v\in\partial\Gamma_i$. Since $\Gamma_i$ is convex, we can always choose these boundary conditions to be integer-valued, and then there exist (infinitely many) integer-valued solutions. As can easily be seen in the following, any such solution results in the same outcome, since the difference $X_i^\alpha-X_i^\beta$ of any two possible choices $X_i^\alpha$ and $X_i^\beta$ is integer-valued harmonic.

Given the functions $X_i$, we define $H_B\in\mathbb{Q}^\Gamma$ by
\begin{align*}
H_B=\hat{H}_B+\sum_i X_i.
\end{align*}
By construction, $H_B$ is harmonic everywhere and has an integer-valued Laplacian. We can thus reinterpret $H_B$ to be an element of $\mathcal{H}_\mathbb{G}^{\Gamma_B}$. It is then easy to see that the function $\xi:\mathcal{H}_G^{P_A}\rightarrow\mathcal{H}_G^{P_B}$, $\xi(H_A)=H_B$, is injective, and that it satisfies $\xi(H_B^1+H_B^2)=\xi(H_B^1)+\xi(H_B^2)$. The function $\xi$ is thus a group monomorphism, and with the isomorphism $-\Delta_{\Gamma}:\mathcal{H}_G^{P}\cong G_P$
from Lemma~\ref{lemma:exactSequence}, we get that $\mu(T^{P_A\rightarrow P_B})=\Delta_{\Gamma_B}\circ\xi\circ(\Delta_{\Gamma_A})^{-1}:G_{P_A}\rightarrowtail G_{P_B}$ is the group monomorphism which we claimed to exist in Theorem~\ref{theorem:main}.

\section{A basis for integer-valued harmonic functions}\label{section:basis}
\begin{figure}[ht!]
	\centering
	 	\includegraphics[width=0.65\textwidth]{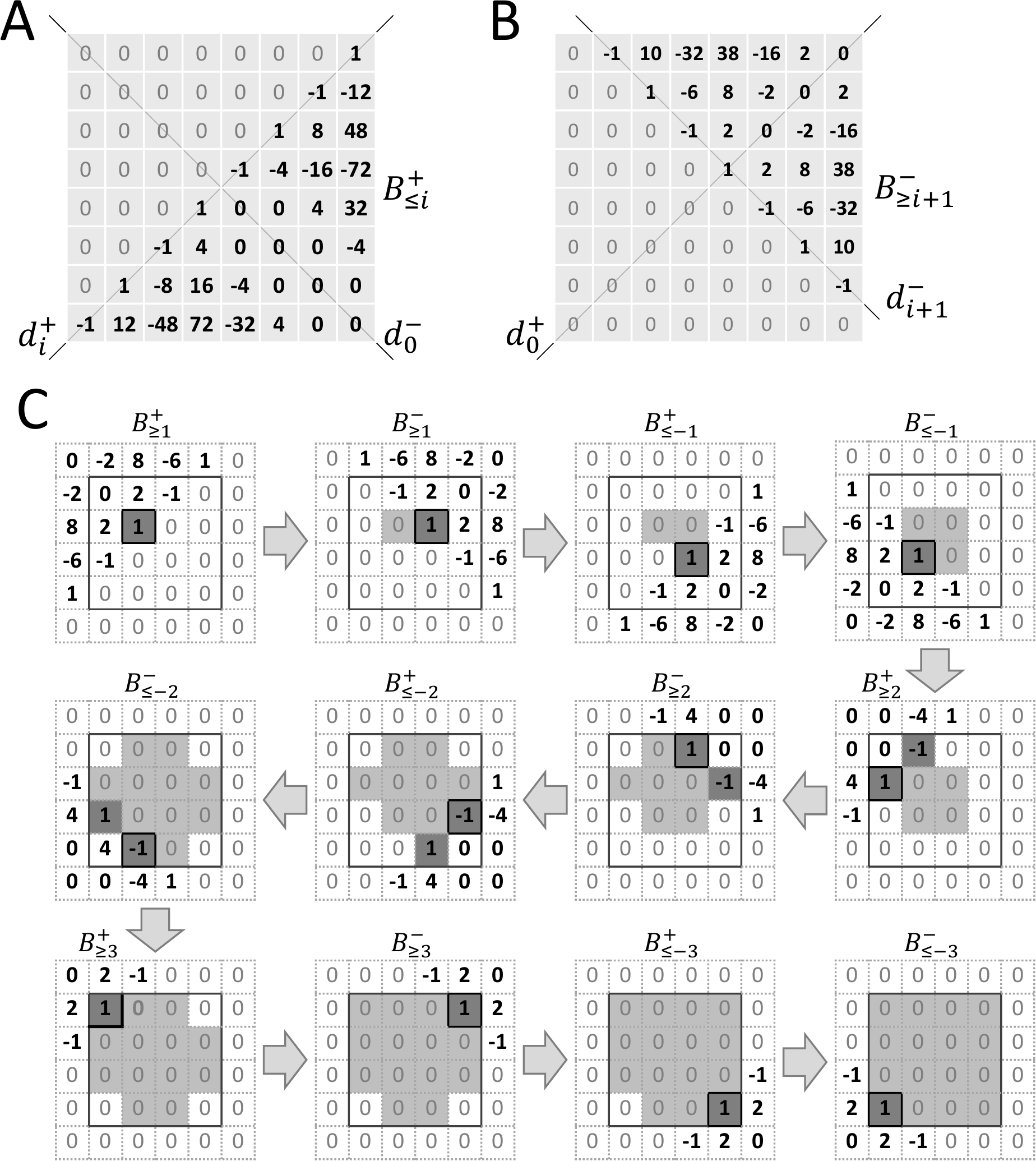}
	\caption{Harmonic functions used in the construction of a basis for $\mathcal{H}^\Gamma_\mathbb{Z}$ (A\&B), and construction of such a basis on a $4\times 4$ square domain (C). 
	A\&B) Two examples of harmonic functions which are zero for all vertices above/below some diagonal (A: $d^+_i$, and B: $d^-_{i+1}$), and which take values in $\{+1,-1\}$ on the diagonal. Both harmonic functions are additionally chosen such that they are zero on the diagonal $d^-_0$ (A), respectively $d^+_0$ (B), which is orthogonal to their respective defining diagonal. The rest of the values of the harmonic functions are chosen such that they are anti-symmetric (A), respectively symmetric (B) with respect to this orthogonal diagonal, depending if the two diagonals intersect on a vertex (B) or not (A).
	C) Each square depicts a step in the construction of the basis by the algorithm in Figure~\ref{fig:basisAlgo}. Light-gray backgrounds denote those vertices already belonging to the growing domain at the beginning of the step, and dark-gray backgrounds those vertices added during the step. The numbers correspond to the harmonic basis function added during the respective step. 
	}
	\label{fig:properBasis}
\end{figure}
In this section, we present an algorithm for the construction of a basis for the module of integer-valued harmonic functions $\mathcal{H}^\Gamma_\mathbb{Z}$ on a finite convex domain $\Gamma\subset\mathbb{Z}^2$. We note that the resulting bases provided important intuition during the research which lead to this article. 

We first define four families of integer-valued harmonic functions from which our algorithm will then select a subset for the construction of the basis.
Denote by $d^+_i=\{(x,y)\in\mathbb{Z}^2|x+y=i+c^+\}$ and $d^-_i=\{(x,y)\in\mathbb{Z}^2|x-y=i+c^-\}$ the diagonals of $\mathbb{Z}^2$, with $c^+,c^-\in\mathbb{Z}$. For a given diagonal, say $d_i^+$, it is then possible to construct an integer-valued harmonic function $\Hag{i}\in\mathcal{H}^\Gamma_\mathbb{Z}$ which is zero for all vertices below $d_i^+$ (i.e. for all $v\in d_j^+$, $j<i$), but non-zero for nearly all vertices on and above $d_i^+$ (i.e. for $v\in d_j^+$, $j\geq i$, see \cite{Buhovsky2017}). For one vertex on each of the non-zero diagonals $d_j^+$, $j\geq i$, the value of $\Hag{i}$ can be freely assigned, which then uniquely determines the value of all other vertices \cite{Buhovsky2017}. 
Here, we only assume that the free value of $\Hag{i}$ on the defining diagonal $d_i^+$ is chosen to be $\pm 1$, while we yet do not pose any restrictions on the choice of the free values on the other diagonals. This construction results in a harmonic function $\Hag{i}$ which is zero below $d_i$ and which alternates between $+1$ and $-1$ on $d_i^+$ (Figure~\ref{fig:properBasis}A\&B). 
Similarly, we denote by $\Hal{i}\in\mathcal{H}^\Gamma_\mathbb{Z}$ an harmonic function which is zero \textit{above} $d_i^+$ and takes the values $\pm 1$ on $d_i^+$, and by $\Hbg{i}$ and $\Hbl{i}$ the corresponding harmonic functions when replacing $d_i^+$ by $d_i^-$ in the definitions above.

We make the following definitions:
\begin{definition}[Diamond Hull]
Let $\Gamma\subset\mathbb{Z}^2$ be a finite convex domain. Denote by $\ext(\Gamma)=\{v\in\mathbb{Z}^2|v\in\Gamma\vee\exists w\in\Gamma: w\sim v\wedge\sum_{u\in\Gamma}\delta_{u\sim w}=3\}$ the domain obtained by extending $\Gamma$ by all vertices in its complement $\mathbb{Z}^2\setminus\Gamma$ which are the direct neighbors of vertices in $\Gamma$ which already have three neighbors in $\Gamma$. With $\ext^k=\ext^{k-1}\circ\ext$ and $\ext^0(\Gamma)=\Gamma$, we then define the diamond hull of $\Gamma$ as the limit $\diam(\Gamma)=\lim_{k\rightarrow\infty}\ext^k(\Gamma)$.
\end{definition}
\begin{corollary}
For every finite convex domain $\Gamma\subset\mathbb{Z}^2$, $\diam(\Gamma)\subset\mathbb{Z}^2$ is a finite convex domain, too. Furthermore, $|\partial\diam(\Gamma)|=|\partial\Gamma|$.
\end{corollary}
\begin{corollary}\label{ext2diam}
Every harmonic function $H\in\mathcal{H}^\Gamma_R$, $R\in\{\mathbb{Z},\mathbb{Q},\mathbb{R}\}$, on a finite convex domain $\Gamma\subset\mathbb{Z}^2$ can be uniquely extended to the diamond hull of $\Gamma$, i.e. there exists a unique $\hat{H}\in\mathcal{H}^{\diam(\Gamma)}_R$ such that $\hat{H}|_\Gamma=H$. The domain $\diam(\Gamma)$ is maximal with respect to this property.
\end{corollary}
It directly follows that, if $\mathcal{B}^{\diam(\Gamma)}_R$ denotes the result of extending all basis functions in $\mathcal{B}^\Gamma_R$ to $\diam(\Gamma)$, then $\mathcal{B}^{\diam(\Gamma)}_R$ is a basis for $\mathcal{H}^{\diam(\Gamma)}_R$. 
At least for the bases constructed below, the reverse is also true.
\begin{definition}[Line-segment]
A vertex $v\in\Gamma$ is a line-segment in $\Gamma\subset\mathbb{Z}^2$ if it has exactly two neighbors $w_1$ and $w_2$ in $\Gamma$, and if $v$, $w_1$ and $w_2$ lie on a line. We denote by $\linecenter(\Gamma)\subset\Gamma$ the set of all line-segments in $\Gamma$.
\end{definition}

With these definitions, our algorithm for the construction of a basis for the module $\mathcal{H}^\Gamma_\mathbb{Z}$ of integer-valued harmonic functions is given in Figure~\ref{fig:basisAlgo}, and exemplified in Figure~\ref{fig:properBasis}C.
\begin{figure}
\begin{algorithm}[H]
\DontPrintSemicolon
\KwIn{A finite convex domain $\Gamma\subset\mathbb{Z}^2$.}
\KwOut{A basis $\mathcal{B}^\Gamma_\mathbb{Z}$ for $\mathcal{H}^\Gamma_\mathbb{Z}$.}
\Begin{
	Set $\Gamma_0$:=$\{\}$, $\mathcal{B}_0$:=$\{\}$, $s$:=$0$\;
	\While{$\Gamma_s\neq\Gamma$}
	 {
		$s$:=$s+1$\;
		Choose $v_s\in\Gamma\setminus \Gamma_{s-1}$ such that $\Gamma_{s-1}\cup\{v_s\}$ is convex and $\linecenter(\Gamma_{s-1}\cup\{v_s\})\subseteq\linecenter(\Gamma)$\;
		Determine $i$ and $j$ such that $d^+_i\cap d^-_j=\{v_s\}$\;
		Choose $B_s\in\{\Hag{i}, \Hal{i}, \Hbg{j}, \Hbl{j}\}$ such that $\left.B_s\right|_{\Gamma_{s-1}}=0$\;
		Set $\mathcal{B}_s$ := $\mathcal{B}_{s-1}\cup\{B_s|_\Gamma\}$\;
		Set $\Gamma_s=\diam(\Gamma_{s-1}\cup\{v_s\})\cap\Gamma$\;
	 }
	\KwRet{$\mathcal{B}_s$}
}
\end{algorithm} 
\caption{An algorithm for the construction of a basis for the module $\mathcal{H}^\Gamma_\mathbb{Z}$ of integer-valued harmonic functions on a finite convex domain $\Gamma\subset\mathbb{Z}^2$.}
\label{fig:basisAlgo}
\end{figure}

\begin{lemma}\label{algoBasisLemma}
For every finite convex domain $\Gamma\subset\mathbb{Z}^2$, the algorithm terminates and returns a basis $\mathcal{B}^\Gamma_\mathbb{Z}$ for the module $\mathcal{H}^\Gamma_\mathbb{Z}$.
\end{lemma}
\begin{proof}
It is easy to see that, in every step $s$, there always exist at least one vertex $v_s$ such that $\Gamma_{s-1}\cup\{v_s\}$ is convex and $\linecenter(\Gamma_{s-1}\cup\{v_s\})\subseteq\linecenter(\Gamma)$. To prove termination of the algorithm, we thus only have to show that, independent of the choice of $v_s$, at least one of the harmonic functions $\Hag{i}, \Hal{i}, \Hbg{j}, \Hbl{j}$ is zero on $\Gamma_{s-1}$. For $s=1$, this is trivially true. For $s>1$, $v_s$ has at least one neighbor in $\Gamma_{s-1}$. Denote this vertex by $v_N$, and, w.l.o.g., assume that it is to the bottom of $v_s$. In $\Gamma_{s-1}$, $v_N$ must have at most two neighbors, since otherwise $v_s\in\diam(\Gamma_{s-1})$. Furthermore, $v_N$ must not have both a neighbor to the right and to the left in $\Gamma_{s-1}$, since this would constitute a line segment in $\Gamma_{s-1}$ which is not a line-segment in $\Gamma$. W.l.o.g., assume that the vertex to the right of $v_N$ is not in $\Gamma_{s-1}$, and denote this vertex by $v_R$. Furthermore, let $d^+_i$ be the diagonal going through $v_s$ and $v_R$. No vertex on this diagonal (and thus also not to the right of this diagonal) can be an element of $\Gamma_{s-1}$, since otherwise either $\Gamma_{s-1}\cup\{v_s\}$ would not be convex, $\linecenter(\Gamma_{s-1})\not\subseteq\linecenter(\Gamma)$, or $v_s\in\diam(\Gamma_{s-1})$. 
Thus, $B_s=\Hag{i}$ is a valid choice at step $s$. 
To show that $\mathcal{B}^\Gamma_\mathbb{Z}$ is a basis for $\mathcal{H}^\Gamma_\mathbb{Z}$, note that $B_1$ is the only harmonic function in $\mathcal{B}^\Gamma_\mathbb{Z}$ which is non-zero at $v_1$. By definition, $B_1$ takes the value $\pm 1$ at $v_1$, and thus only linear combination of $\mathcal{B}^\Gamma_\mathbb{Z}$ can be integer-valued for which the coefficient corresponding to $B_1$ is integer-valued. Assume that, in step $s$, this is true for all harmonic functions in $\mathcal{B}_{s-1}$. Then, since $B_s$ takes the value $\pm 1$ at $v_s$, this is also true for $B_s$. 
Thus, the functions in $\mathcal{B}^\Gamma_\mathbb{Z}$ are linearly independent. Since $|\mathcal{B}_s|=|\partial\Gamma_s|=s$, $|\mathcal{B}^\Gamma_\mathbb{Z}|=|\partial\Gamma|$. Thus, $\mathcal{B}^\Gamma_\mathbb{Z}$ is a basis for $\mathcal{H}^\Gamma_\mathbb{R}$. We conclude our proof by noting that $\mathcal{H}^\Gamma_\mathbb{Z}\subset\mathcal{H}^\Gamma_\mathbb{R}$.
\end{proof}

\begin{corollary}\label{corollary:properSquareBasis}
Let $\Gamma_N\subset\mathbb{Z}^2$ be an $N\times N$ square domain. Assume that $d^+_0$ and $d^-_0$ correspond to the main diagonals of the domain. Then, (i) for $N=1$, $\{\Hag{0}\}$ is a basis for $\mathcal{H}^\Gamma_\mathbb{Z}$; (ii) for $N\in 2\mathbb{N}$, $\{\Hag{i},\allowbreak\Hal{-i},\allowbreak\Hbg{i},\allowbreak\Hbl{-i}\}_{i=1,\ldots,N-1}$ is a basis; and (iii) for $N\in 2\mathbb{N}+1$, $\{\Hag{0},\allowbreak\Hal{-1},\allowbreak\Hbg{1},\allowbreak\Hbl{-1}\} \cup\allowbreak\{\Hag{i},\allowbreak\Hal{-i},\allowbreak\Hbg{i},\allowbreak\Hbl{-i}\}_{i=2,\ldots,N-1}$ is a basis for $\mathcal{H}^\Gamma_\mathbb{Z}$.
\end{corollary}
\begin{proof}
See construction in Figure~\ref{fig:properBasis}C.
\end{proof}

\section{A harmonic function related to primes}
\begin{figure}[htb!]
	\centering
	 	\includegraphics[width=0.75\textwidth]{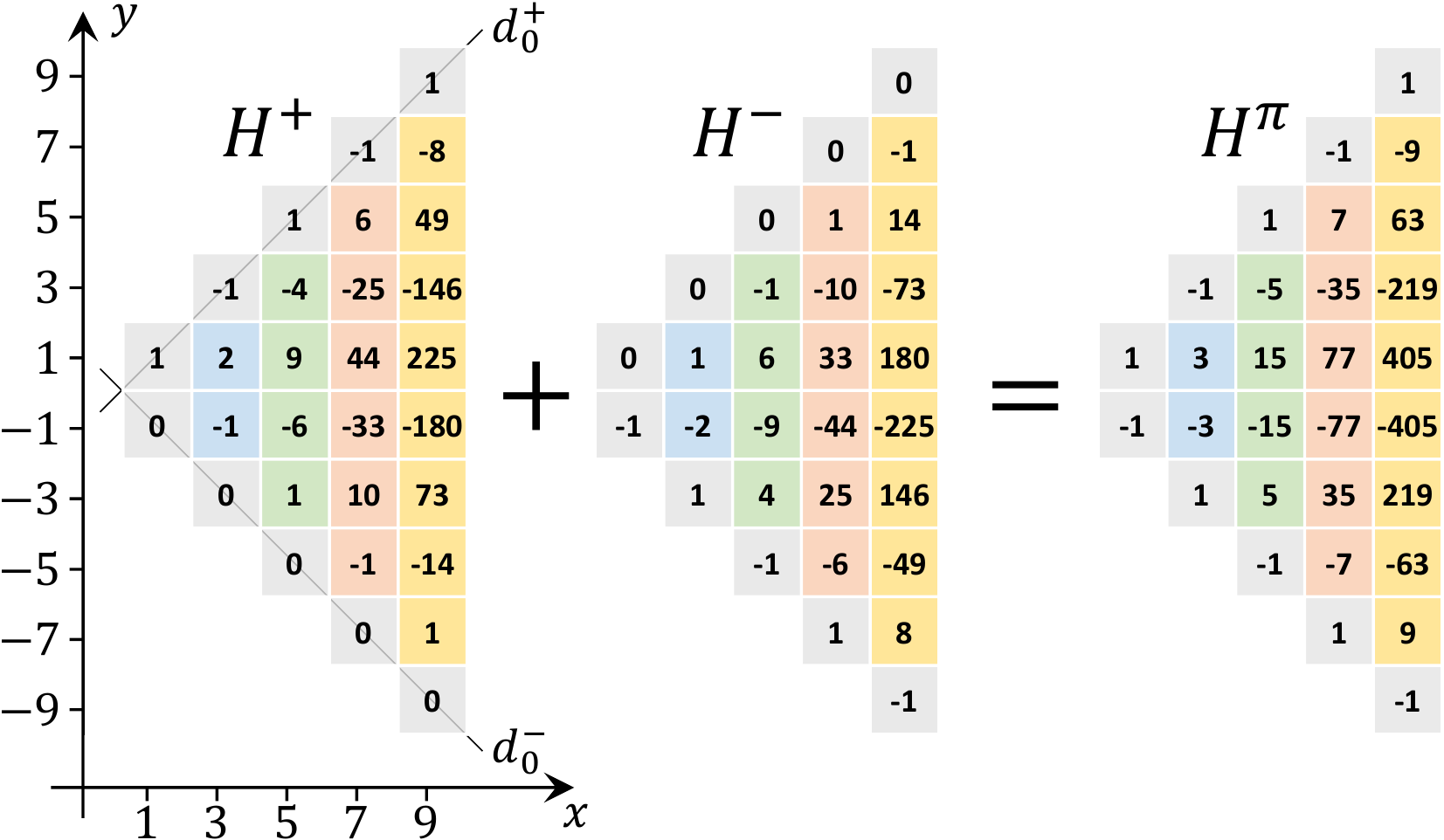}
	\caption{In each quadrant of the domain, $H^{\pi}=H^++H^-$ is the sum of two harmonic functions $H^+$ and $H^-$. Our proof of Lemma~\ref{lemma:divPrime} is based on showing that $(x-y)H^+(x,y)=-(x+y)H^+(x,-y)$.
	}
	\label{fig:primeConstruction}
\end{figure}
To prove Lemma~\ref{lemma:divPrime}, we first observe that, for $N=1$, $G^\Gamma\cong\mathbb{Z}/4\mathbb{Z}\supseteq\mathbb{Z}/2\mathbb{Z}$, in agreement with the lemma. Since $N+1$ must be prime, we thus assume that $N\in 2\mathbb{N}$ in the following.

Define the diagonals $d^+_0$ and $d^-_0$ (see previous section) such that they do not intersect at a vertex, and choose the harmonic basis functions $\Hag{2i-1}$, $\Hal{-2i+1}$,$\Hbg{2i-1}$, and $\Hbl{-2i+1}$, $i\in\mathbb{N}$, as described in Figure~\ref{fig:properBasis}B. The harmonic function depicted in Figure~\ref{fig:primeHarmonic}B is then given by
\begin{align*}
H^\pi=\sum_{i=0}^\infty(-1)^{i}(\Hag{2i+1}+\Hal{-(2i+1)}-\Hbg{2i+1}-\Hbl{-(2i+1)}).
\end{align*}

By definition, $H^\pi$ is symmetric with respect to $d^+_0$ and $d^-_0$. Thus, we w.l.o.g. only consider the quadrant of the domain below $d^+_0$ and above $d^-_0$, and label the vertices in this quadrant by coordinates $(x,y)$, with $x=1,3,5,\ldots$ and $y=\pm 1,\pm 3,\ldots\pm x$, as shown in Figure~\ref{fig:primeConstruction}. In this quadrant, $H^\pi=H^++H^-$ with $H^+=\sum_{i=0}^\infty(-1)^{i}\Hag{2i+1}$, and $H^-=\sum_{i=0}^\infty(-1)^{i+1}\Hbg{2i+1}$ (Figure~\ref{fig:primeConstruction}). We claim that, in this quadrant, $(x-y)H^+(x,y)=-(x+y)H^+(x,-y)$. Since $H^+(x,y)=-H^-(x,-y)$, this implies that 
\begin{align*}
H^\pi=&H^+(x,y)+H^-(x,y)
=H^+(x,y)-H^+(x,-y)\\
=&H^+(x,y)+\frac{x-y}{x+y}H^+(x,y)
=\frac{2x}{x+y}H^+(x,y).
\end{align*}
For $|y|<x$, since $y\neq 0$ and $H^+(x,y)$ is integer-valued, $H^\pi(x,y)$ is thus divisible by $x$ given that $x$ is prime. Thus, if our claim holds, Lemma~\ref{lemma:divPrime} directly follows from the discussion in the Introduction when setting $x=N+1$. 

For small enough values of $x$, it is possible to directly check our claim for all $|y|\leq x$ (Figure~\ref{fig:primeHarmonic}). 
Furthermore, it is easy to validate (by induction, in this order) that $H^+(x,x)=\pm 1$, $H^+(x,-x)=0$, $H^+(x,x-2)=\mp (x-1)$, $H^+(x,-x+2)=\pm 1$, $H^+(x,x-4)=\pm(x-2)^2$, and $H^+(x,-x+4)=\mp (2x-4)$. Thus, our claim also holds close to the diagonals $d_0^+$ and $d_0^-$.

Now, assume that our claim holds for all $\hat{x}\leq x$. Then,
{\small
\begin{align*}
H^+(x+2,y)=&4H^+(x,y)-H^+(x,y-2)-H^+(x,y+2)-H^+(x-2,y)\\
=&-4\frac{x+y}{x-y}H^+(x,-y)+\frac{x+y-2}{x-y+2}H^+(x,-y+2)\\
&+\frac{x+y+2}{x-y-2}H^+(x,-y-2)+\frac{x+y-2}{x-y-2}H^+(x-2,-y)\\
=&-\frac{x+y+2}{x-y+2}H^+(x+2,-y)+\frac{4}{x-y+2}\epsilon(x,-y),
\end{align*}
}
with
{\small
\begin{align*}
\epsilon(x,y) =& -\frac{4y}{x-y}H^+(x,-y)-H^+(x,-y+2)\\
&+\frac{x+y+2}{x-y-2}H^+(x,-y-2)+\frac{2y}{x-y-2}H^+(x-2,-y)\\
=& \frac{4y}{x+y}H^+(x,y)+\frac{x-y+2}{x+y-2}H^+(x,y-2)\\
&-H^+(x,y+2)-\frac{2y}{x+y-2}H^+(x-2,y).
\end{align*}
}
The following calculations, which show that $\epsilon(x,y)=0$, are ``a bit tedious'' to check by hand, and we thus recommend using a computer algebra system. We first utilize that $H^+$ is harmonic to replace $H^+(x,y)$ by $4H^+(x-2,y)-H^+(x-2,y+2)-H^+(x-2,y-2)-H^+(x-4,y)$, and similarly for $H^+(x,y-2)$ and $H^+(x,y-2)$. We then utilize that, by the inductive assumption, $\epsilon(\hat{x},y)=0$ for all $\hat{x}<x$ to replace $H^+(x-2,y+4)$ by $\frac{4y+8}{x+y}H^+(x-2,y+2)+\frac{x-y-2}{x+y-2}H^+(x-2,y)-\frac{2y+4}{x+y-2}H^+(x-4,y+2)$, and similarly for $H^+(x-2,y+2)$ and $H^+(x-2,y-4)$, which also eliminates $H^+(x-2,y-2)$. The only remaining term in column $x-2$ is then in $H^+(x-2,y)$, which we replace again by $4H^+(x-4,y)-H^+(x-4,y+2)-H^+(x-4,y-2)-H^+(x-6,y)$. We then get that $\epsilon(x,y)=-\frac{x+y-6}{x+y-2}\epsilon(x-4,y)=0$, which concludes the proof.

\section{Diamond-shaped harmonic functions}
From the discussion in the Introduction, it becomes clear that every harmonic function of the type depicted in Figure~\ref{fig:primeHarmonic}C directly corresponds to a cyclic subgroup $\mathbb{Z}\setminus 4\mathbb{Z}\subseteq G_\Gamma$ of the sandpile group $G_\Gamma$ with order four.
To prove Lemma~\ref{lemma:div4}, we thus only have to show that there exist $N$ such harmonic functions which are linearly independent. Since, for $N=1$, the lemma is trivially satisfied ($S_1=\mathbb{Z}/4\mathbb{Z}\cong G^\Gamma$), we assume $N>1$ in the following.

For $N\in 2\mathbb{N}$, define $d^+_0$ and $d^-_0$ such that they correspond to the main diagonals of the domain. Then, by Corollary~\ref{corollary:properSquareBasis}, $\mathcal{B}^\Gamma_\mathbb{Z}=\{\Hag{i}|_\Gamma,\allowbreak\Hal{-i}|_\Gamma,\allowbreak\Hbg{i}|_\Gamma,\allowbreak\Hbl{-i}|_\Gamma\}_{i=1,\ldots,N-1}$ is a basis for the $\mathcal{H}_\mathbb{Z}^\Gamma$ on $\Gamma$. 
Recall, that $\Hag{i}$ is defined on the whole of $\mathbb{Z}^2$, takes values in $\{-1,+1\}$ on its defining diagonal $d^+_i$, and is zero below it (and similar for $\Hal{-i}$, $\Hbg{i}$, and $\Hbl{-i}$). Also recall that on every diagonal $d_j^+$, $j>i$, we can still choose the value of one vertex, which then determines the values $\Hag{i}$ of all other vertices on the same diagonal \cite{Buhovsky2017}. If we always choose these ``free values'' to be zero, $\Hag{i}$ becomes divisible by four everywhere except on its defining diagonal $d_i^+$ (Figure~\ref{fig:primeHarmonic}C). 

On the boundary $\partial\Gamma^{op}$ of the complement $\Gamma^{op}=\mathbb{Z}^2\setminus\Gamma$ of the domain, this implies that $\Hag{i}$ is divisible by four, except for two vertices for which $\Hag{i}$ takes a value of $\pm 1$ (Figure~\ref{fig:primeHarmonic}A). 
Furthermore, for each vertex $v\in\partial\Gamma^{op}$ on this boundary, there exist exactly two basis functions $B_i, B_j\in\mathcal{B}^{\Gamma}_\mathbb{Z}$ with $B_i(v), B_j(v)\in\{-1, +1\}$. Thus, for each $i=1,\ldots,N$, we can define the harmonic function $H^\diamond_i=\Hag{i}\pm\Hbg{N-i+1}\pm\Hal{-i}\pm\Hbl{-N+i-1}$, where the signs are chosen such that the values of the basis functions on the boundary which are $\pm 1$ cancel each other out (Figure~\ref{fig:primeHarmonic}C). It then directly follows that $\Delta_\Gamma (H^\diamond_i|_\Gamma)$ is divisible by four.

We can always replace one basis function in a basis for $\mathcal{H}_\mathbb{Z}^\Gamma$ by the sum of it and an integer-multiple of another function in the same basis. That is, if $\{B_1,\ldots,B_i,\ldots B_{|\partial\Gamma|}\}$ is a basis for $\mathcal{H}_\mathbb{Z}^\Gamma$, $\{B_1,\ldots,B_i+z B_j,\ldots B_{|\partial\Gamma|}\}$ is so, too, for every $z\in\mathbb{Z}$ and $i\neq j$. Trivially, since $\Hag{N}$, $\Hal{-N}$,$\Hbg{N}$ and $\Hbl{-N}$ evaluate to zero in $\Gamma$, we can also add their restriction to $\Gamma$ to any basis functions.
Together, this means that $\{H^\diamond_1|_\Gamma,\allowbreak\Hal{-1}|_\Gamma,\allowbreak H^\diamond_N|_\Gamma,\allowbreak\Hbl{-1}|_\Gamma\}\cup\{H^\diamond_i|_\Gamma,\allowbreak\Hal{-i}|_\Gamma,\allowbreak\Hbg{i}|_\Gamma,\allowbreak\Hbl{-i}|_\Gamma\}_{i=2,\ldots,N-1}$ is also a basis for $\mathcal{H}_\mathbb{Z}^\Gamma$, from which Lemma~\ref{lemma:div4} directly follows.

For $N\in 2\mathbb{N}+1$, by Corollary~\ref{corollary:properSquareBasis}, $\{\Hag{0}|_\Gamma,\allowbreak\Hal{-1}|_\Gamma,\allowbreak\Hbg{1}|_\Gamma,\allowbreak\Hbl{-1}|_\Gamma\} \cup\allowbreak\{\Hag{i}|_\Gamma,\allowbreak\Hal{-i}|_\Gamma,\allowbreak\Hbg{i}|_\Gamma,\allowbreak\Hbl{-i}|_\Gamma\}_{i=2,\ldots,N-1}$ is a basis for $\mathcal{H}^\Gamma_\mathbb{Z}$. By a similar argument as before, we get that also $\{\Hag{0}|_\Gamma,\allowbreak\Hal{-1}|_\Gamma,\allowbreak H^\diamond_N|_\Gamma,\allowbreak\Hbl{-1}|_\Gamma\} \cup\allowbreak\{H^\diamond_i|_\Gamma,\allowbreak\Hal{-i}|_\Gamma,\allowbreak\Hbg{i}|_\Gamma,\allowbreak\Hbl{-i}|_\Gamma\}_{i=2,\ldots,N-1}$ is a basis. Note that, different to before, this basis only contains $N-1$ diamond shaped basis functions $H^\diamond_i$. However, the Laplacian of $\Hag{0}$ is divisible by four, too, which concludes our proof of Lemma~\ref{lemma:div4}.

\bibliographystyle{aomplain}
\bibliography{Lang2019}
\end{document}